\documentclass[
    aps,
    pra,
    twocolumn,
    10pt,
    nofootinbib,
    superscriptaddress
]{revtex4-2}

\usepackage{amsmath}
\usepackage{amssymb}
\usepackage{amsthm}
\usepackage{braket}
\usepackage{bm}
\usepackage{mathtools}
\usepackage{microtype}
\usepackage{graphicx}
\usepackage[dvipsnames]{xcolor}

\usepackage[hidelinks]{hyperref}
\usepackage{tikz}
\usepackage[T1]{fontenc}
\usetikzlibrary{calc}

\newtheorem{theorem}{Theorem}

\DeclareMathOperator{\SR}{SR}

\newcommand{\swap}{\mathcal{F}}

\begin{document}

\title{Two-copy distillability of one-copy-undistillable negative-partial-transpose states in every dimension}

\author{Gelo Noel M. Tabia}
\email{gelo.tabia@foxconn.com}
\affiliation{Hon Hai (Foxconn) Research Institute, Taipei 114, Taiwan}

\author{Kai-Siang Chen}
%\email{kschen@phys.ncku.edu.tw}
\affiliation{Department of Physics and Center for Quantum Frontiers of Research \& Technology (QFort), National Cheng Kung University, Tainan 701, Taiwan}

\author{Min-Hsiu Hsieh}
%\email{min-hsiu.hsieh@foxconn.com}
\affiliation{Hon Hai (Foxconn) Research Institute, Taipei 114, Taiwan}

\date{\today}
\begin{abstract}
Whether entanglement with a negative partial transpose (NPT) can be
undistillable is a longstanding open problem. In the canonical family
of DiVincenzo \textit{et al.}, the one-copy-undistillable region was
conjectured to remain undistillable at all copy numbers. We refute this:
in every dimension $d\geq3$, one vertex of this region is two-copy
distillable, via an explicit Schmidt-rank-two tight-frame witness.
The separation persists on an open set, while other states in the
same region are provably two-copy undistillable. Twirling our
counterexample into Werner form makes it two-copy undistillable, so
this standard reduction can erase finite-copy distillability.
\end{abstract}

\maketitle

\section{Introduction}

A bipartite state with negative partial transpose (NPT) is necessarily
entangled~\cite{Peres1996}, but its entanglement need not be accessible
to local filtering on a single copy. Whether every NPT state becomes
distillable from some finite number of copies remains one of the
longest-standing open questions in entanglement
theory~\cite{Horodecki1998,Dur2000,Clarisse2006,
HorodeckiReview2009,HorodeckiOpenProblems}. The obstruction is not the
existence of a negative direction of the partial transpose, but whether
such negativity can eventually be accessed by a Schmidt-rank-two
direction across the collective Alice--Bob bipartition. Copy number can
matter arbitrarily strongly: for every prescribed finite number of
copies there exist states that remain undistillable up to that number
but become distillable later~\cite{Watrous2004}. What remains unknown
is whether this obstruction can persist forever.

DiVincenzo \textit{et al.} introduced a canonical two-parameter family
for this question~\cite{DiVincenzo2000}. Local filtering
followed by symmetry twirling maps general NPT states into this family
while preserving a negative partial-transpose witness, making the
resulting plane more than a convenient symmetric example. Within this
plane, DiVincenzo \textit{et al.} identified the region bounded by $BCGK$
as one-copy undistillable, with $BK$ its PPT/NPT boundary. A schematic 
of this plane and its basic geometry for \(d = 3\) is given in 
Fig.~\ref{fig:qutrit-geometry}. In the caption of Fig.~2 of 
Ref.~\cite{DiVincenzo2000}, they explicitly conjectured that the entire 
$BCGK$ region is undistillable by any means. For more than two decades, 
this region therefore served as a canonical candidate setting for 
NPT bound entanglement.

The two-copy problem on the Werner line was studied using numerical,
Schmidt-rank, and matrix-inequality methods~\cite{ViannaDoherty2006,
Pankowski2010,Djokovic2016,Qian2021,Qi2024,CostaRico2025,
CostaRicoWolf2025}.
Recent independent results have established that Werner states are
two-copy distillable exactly when they are already one-copy
distillable~\cite{WernerTwoCopyFu,SongChen2026,
WernerTwoCopyFraser,WernerTwoCopyBharti}. Together with the propagation
lemma of DiVincenzo \textit{et al.}, this establishes two-copy
undistillability throughout the triangle $BGK$. In the present work we
prove the opposite behavior on the other side of $BG$: point $C$ is
one-copy undistillable but two-copy distillable for every local
dimension $d\geq3$. The same explicit witness remains negative in a
relative neighborhood of $C$ within the physical canonical plane, and
the separation persists on an open set in the full state space. Using
the Werner two-copy theorem, we also construct exact certificates that
extend rigorous two-copy undistillability from the $BG$ side into
$BCG$ for $d=3,\ldots,8$. The canonical one-copy-undistillable region
therefore contains rigorously established subregions with opposite
two-copy behavior.

\begin{figure}[!t]
\centering
\begin{tikzpicture}[
    x=1.55cm,
    y=1.02cm,
    every node/.style={font=\small},
    point/.style={circle,fill=black,inner sep=1.5pt}
]
    % Drawing coordinates:
    % X = 30(b-1/6), Y = 24c.

    \pgfmathsetmacro{\xBp}{5-2*sqrt(5)}
    \pgfmathsetmacro{\xGp}{15/19}
    \pgfmathsetmacro{\yGp}{8/19}

    % Exact canonical coordinates for d=3.
    \coordinate (B)  at (0,0);
    \coordinate (C)  at ({5/7},0);
    \coordinate (F)  at (5,0);
    \coordinate (K)  at (0,4);
    \coordinate (G)  at (1,{8/5});
    \coordinate (H)  at (0,2);
    \coordinate (Bp) at (\xBp,0);
    \coordinate (Gp) at (\xGp,\yGp);

    % Schematic meeting point of the additional certified
    % two-copy-undistillable region with BC.
    \coordinate (E) at ({35-20*sqrt(3)},0);

    % ------------------------------------------------------------
    % Filled regions
    % ------------------------------------------------------------

    % Unclassified part of BCG.
    \fill[gray!12]
        (B)--(C)--(G)--cycle;

    % BGK: two-copy undistillable from the Werner two-copy theorem
    % together with the propagation argument.
    \fill[ForestGreen!18]
        (B)--(G)--(K)--cycle;

    % Additional exact two-copy-undistillable region inside BCG,
    % shown schematically.
    \fill[ForestGreen!32]
        (B)--(G)
        .. controls (0.65,1.04) and (0.47,0.45) ..
        (E)--cycle;

    \draw[ForestGreen!70!black,thick]
        (G)
        .. controls (0.65,1.04) and (0.47,0.45) ..
        (E);

    % Exact fixed-witness two-copy-distillable lobe.
    \path[fill=blue!12]
        (Bp)--(C)--(Gp)
        plot[
            domain=\xGp:\xBp,
            samples=160,
            smooth,
            variable=\u
        ]
        ({\u},{(8*sqrt(\u*(11*\u+30))-28*\u-20)/5})
        --cycle;

    % ------------------------------------------------------------
    % Canonical boundaries
    % ------------------------------------------------------------

    \draw[thick]
        (B)--(F)--(K)--cycle;

    \draw[thick]
        (B)--(C)--(G)--(K)--cycle;

    % Keep BG visible through the green shading.
    \draw[dash dot]
        (B)--(G);

    % Werner line and one-copy-undistillable segment.
    \draw[dashed,thick]
        (F)--(H);

    \draw[very thick]
        (G)--(H);

    % Exact fixed-witness boundary.
    \draw[blue!70!black,thick]
        plot[
            domain=\xBp:\xGp,
            samples=160,
            smooth,
            variable=\u
        ]
        ({\u},{(8*sqrt(\u*(11*\u+30))-28*\u-20)/5});

    % ------------------------------------------------------------
    % Canonical points
    % ------------------------------------------------------------

    \node[
        point,
        label=below left:$B$
    ] at (B) {};

    \node[
        point,
        label=below:$C$
    ] at (C) {};

    \node[
        point,
        label=below:$F$
    ] at (F) {};

    \node[
        point,
        label=left:$K$
    ] at (K) {};

    \node[
        point,
        label={[yshift=3pt]right:$G$}
    ] at (G) {};

    \node[
        point,
        label=left:$H$
    ] at (H) {};

    % Orientation label.
    \node[
        rotate=-15,
        font=\footnotesize
    ] at (2.4,1.30)
        {Werner line};

\end{tikzpicture}
\caption{
NPT sector of the canonical two-parameter plane for $d=3$.
The line $BK$ is the PPT/NPT boundary and $FH$ is the Werner line.
The quadrilateral $BCGK$ is one-copy undistillable.
Recent Werner-state results together with the propagation lemma of
Ref.~\cite{DiVincenzo2000} establish two-copy undistillability
throughout the triangle $BGK$, shown in light green.
The darker green region shows schematically the additional
two-copy-undistillable region inside $BCG$ established by the exact
certificates described in the Appendix.
The fixed witness $\ket{\Psi_3}$ constructed here certifies the blue
two-copy-distillable lobe. The remaining intermediate region inside
$BCG$ is not classified.
}
\label{fig:qutrit-geometry}
\end{figure}

Ref.~\cite{DiVincenzo2000} already showed that a further Werner twirl
is sufficient for the all-copy existence question. At finite copy number,
however, this reduction can lose information: point $C$ is two-copy
distillable even though its Werner image is two-copy undistillable. 
The remaining problem is to determine the true two-copy boundary in 
the canonical plane and how the distillable regions change as more copies 
are allowed. Our point-$C$-tailored three-copy witnesses already enlarge 
the certified distillable region in some directions.

At point $C$, the partial transpose separates into a positive diagonal
contribution and a coherent maximally entangled contribution. Opposite
local copy parities give the coherent terms the
favorable signs, while an equal-norm tight frame distributes the
unavoidable positive diagonal cost uniformly over the available basis
labels. Together these give an explicit dimension-uniform construction of 
a Schmidt-rank-two witness for every $d\geq3$. The same parity-and-frame
construction may also be useful for higher-copy certificates and for
states away from point $C$.

Because $(\rho^\Gamma)^{\otimes k}$ depends nonlinearly on $\rho$,
neither the fixed witness $\ket{\Psi_d}$ nor the Werner-side
propagation argument alone determines the intermediate region. The
complementary exact certificates in the Appendix narrow this gap, but
the finite-copy boundary remains unresolved.

\section{Preliminaries}
\label{sec:prelim}

\subsection{Finite-copy distillability}

A bipartite state $\rho$ is $k$-copy distillable if and only if
there exists a vector $\ket{\psi}$ of Schmidt rank at most two
across $A_1\cdots A_k:B_1\cdots B_k$ such that
\begin{equation}
    \bra{\psi}(\rho^\Gamma)^{\otimes k}\ket{\psi}<0.
    \label{eq:k-copy-criterion}
\end{equation}
Here $\Gamma$ denotes partial transposition on Bob's subsystem.
The Schmidt-rank restriction is the operational content of the
criterion: it asks whether local two-dimensional output spaces can
isolate an NPT two-qubit state from the $k$ copies. A negative vector
of the filtered two-qubit partial transpose pulls back to a vector
satisfying Eq.~\eqref{eq:k-copy-criterion}, while conversely the
Schmidt decomposition of $\ket{\psi}$ defines suitable local
rank-two maps~\cite{HorodeckiDistillation,Horodecki1998,Dur2000}.
For one copy, violation of the reduction criterion gives another
standard sufficient condition for distillability~\cite{Horodecki1999_redcrit}.

\subsection{Canonical geometry and the Werner benchmark}

The canonical family is defined by
\begin{align}
    \rho_{b,c}
    ={}&
    a\sum_{i=0}^{d-1}\ket{ii}\!\bra{ii}
    +b\sum_{0\leq i<j\leq d-1}
       \ket{\psi^-_{ij}}\!\bra{\psi^-_{ij}}
    \nonumber\\
    &+
    c\sum_{0\leq i<j\leq d-1}
       \ket{\psi^+_{ij}}\!\bra{\psi^+_{ij}},
    \label{eq:canonical-family}
\end{align}
where $\ket{\psi^\pm_{ij}}=(\ket{ij}\pm\ket{ji})/\sqrt2$ and
\begin{equation}
    da+\frac{d(d-1)}2(b+c)=1.
    \label{eq:canonical-normalization}
\end{equation}
It is physical in the triangle selected by $a,b,c\geq0$. The following
relevant points are listed in $(b,c)$ coordinates:
\begin{align}
    G&=\left(\frac3{d(2d-1)},\frac1{d(2d-1)}\right),
    &
    F&=\left(\frac2{d(d-1)},0\right),
    \nonumber\\
    H&=\left(\frac1{d(d-1)},\frac1{d(d+1)}\right),
    &
    B&=\left(\frac1{d(d-1)},0\right),
    \nonumber\\
    K&=\left(\frac1{d(d-1)},\frac1{d(d-1)}\right),
    &
    C&=\left(\frac4{d(3d-2)},0\right).
    \label{eq:canonical-points}
\end{align}

The quadrilateral $BCGK$ is one-copy undistillable throughout
\cite{DiVincenzo2000}.  The Werner family \cite{Werner1989}
\begin{equation}
    \rho_{\alpha}^{\mathrm W}
    =
    \frac{I+\alpha \swap}{d^2+\alpha d},
    \qquad -1\leq\alpha\leq1,
    \label{eq:werner}
\end{equation}
where $\swap=\sum_{i,j}\ket{ij}\!\bra{ji}$ denotes the swap operator. The Werner 
family forms the line $FH$.  Point $G$ corresponds to $\alpha=-1/2$, the
one-copy-distillability threshold, and $H$ to $\alpha=-1/d$, the 
PPT boundary~\cite{Werner1989,Dur2000}.  Hence $GH$ is the 
one-copy-undistillable Werner segment. Recent results show that 
the one- and two-copy thresholds coincide for Werner states~\cite{WernerTwoCopyFu,SongChen2026,WernerTwoCopyFraser,WernerTwoCopyBharti},
so $G$ and the entire segment $GH$ are two-copy undistillable.
The Werner twirl preserves $\operatorname{Tr}(\rho \swap)$, and for point $C$
its Werner image has
\begin{equation}
    \alpha_C=-\frac{d+2}{3d+2}.
    \label{eq:Werner-image-C}
\end{equation}
For every $d\geq3$,
$-1/2<\alpha_C<-1/d$, so this image lies strictly inside $GH$ and is
two-copy undistillable, whereas we show below that $C$ itself is
two-copy distillable. Thus passing to Werner form can erase two-copy
distillability and need not preserve the number of copies required for
distillation.

DiVincenzo \textit{et al.} also proved a propagation lemma: if $G$ is
$n$-copy undistillable, then every state in $BGK$ is $n$-copy
undistillable~\cite{DiVincenzo2000}.  The added $B$ and $K$ components
are separable, so after partial transposition they decompose into product
projectors.  Projecting some copies of a global Schmidt-rank-two vector 
onto these product states cannot increase the Schmidt rank on the remaining copies.  
Taking $n=2$ and using the recent Werner result therefore establishes that 
all of $BGK$ is two-copy undistillable.

Ref.~\cite{DiVincenzo2000} also showed that, for every fixed $n$,
$n$-copy undistillability persists into a nonzero NPT neighborhood of
the PPT boundary $BK$. Near $B$ this neighborhood crosses $BG$ into
$BCG$; for two copies, the Appendix replaces this qualitative
observation by explicit exact certificates.

\subsection{Why point \texorpdfstring{$C$}{C} is the critical test case}

At point $C$,
\begin{equation}
    (a,b,c)
    =
    \left(
        \frac{1}{3d-2},
        \frac{4}{d(3d-2)},
        0
    \right),
    \label{eq:point-C}
\end{equation}
and we denote the corresponding state by $\rho_{C,d}$.  Introducing
\begin{equation}
    \Delta=\sum_i\ket{ii}\!\bra{ii},
    \label{eq:Delta}
\end{equation}
we have
\begin{align}
    \rho_{C,d}
    &=
    \frac{1}{3d-2}
    \left[\Delta+\frac{2}{d}(I-\swap)\right],
    \label{eq:rhoCd}\\
    \rho_{C,d}^{\Gamma}
    &=
    \frac{1}{3d-2}Q_d,
    \qquad
    Q_d=\Delta+\frac{2}{d}(I-\Omega),
    \label{eq:Qd}
\end{align}
where
\begin{equation}
    \Omega=\sum_{i,j}\ket{ii}\!\bra{jj}
    =d\ket{\Phi_d}\!\bra{\Phi_d},
    \qquad
    \ket{\Phi_d}=\frac1{\sqrt d}\sum_i\ket{ii}.
\end{equation}
Thus
\begin{equation}
    \bra{\Phi_d}\rho_{C,d}^{\Gamma}\ket{\Phi_d}
    =-
    \frac{d-2}{d(3d-2)}<0,
    \label{eq:C-NPT}
\end{equation}
so $C$ is NPT in every dimension $d\geq3$.

Nevertheless, $C$ is one-copy undistillable.  The original proof uses
\begin{equation}
    \rho_{C,d}^{\Gamma}
    =
    \frac{2d-1}{3d-2}\rho_{G,d}^{\Gamma}
    +
    \frac{2}{d(3d-2)}\sum_{i<j}\Pi_{ij},
    \label{eq:C-from-G}
\end{equation}
with
$\Pi_{ij} = \frac12(\ket{ii}-\ket{jj})(\bra{ii}-\bra{jj})\geq0$.
The first term is nonnegative on all Schmidt-rank-at-most-two vectors
because $G$ is one-copy undistillable, while the second term is positive
semidefinite.  The same decomposition does not propagate to several
copies: the ranges of the $\Pi_{ij}$ are entangled, and projecting some copies
onto these entangled subspaces can increase the Schmidt rank on the
remaining copies.
This is the gap in the original argument and motivates testing point
$C$ for a counterexample.

\section{Results}
\label{sec:results}

\subsection{A uniform equal-norm tight-frame certificate}

\begin{theorem}[Two-copy distillability of canonical point
\texorpdfstring{$C$}{C}]
For every local dimension $d\geq3$, the canonical state $\rho_{C,d}$ is
one-copy undistillable but two-copy distillable. Consequently, point $C$
provides, in every dimension $d\geq3$, a counterexample to the conjecture
stated in the caption of Fig.~2 of Ref.~\cite{DiVincenzo2000} that the
entire $BCGK$ region is undistillable.
\label{thm:pointC}
\end{theorem}

The one-copy statement was recalled above. The key feature of the
two-copy construction is that the same structured Schmidt-rank-two
witness works for every $d\geq3$. It is guided by the form
$Q_d=\Delta+\frac{2}{d}(I-\Omega)$ of the partial transpose at $C$.
In the coefficient-matrix representation of a two-copy vector, the
$\Omega$ terms probe coherent sums. Choosing opposite local copy parity
on Alice and Bob removes an unfavorable total-trace contribution,
while equal-norm spreading minimizes the remaining positive diagonal
cost. The coefficient-matrix derivation is given in the Appendix.

Let $n=d-1$ and choose
\begin{align}
    \ket{x_1}
    &=
    \frac1{\sqrt n}
    \sum_{\ell=0}^{n-1}
    \cos\!\left(\frac{\pi\ell}{n}\right)\ket{\ell+1},
    \\
    \ket{x_2}
    &=
    \frac1{\sqrt n}
    \sum_{\ell=0}^{n-1}
    \sin\!\left(\frac{\pi\ell}{n}\right)\ket{\ell+1}.
\end{align}
They obey
\begin{equation}
    \braket{x_r|x_s}=\frac12\delta_{rs},
    \qquad
    \sum_{r=1}^{2}\left|\braket{j|x_r}\right|^2
    =\frac1{d-1}.
    \label{eq:frame-identities}
\end{equation}
The coefficient pairs across the $d-1$ nonzero labels therefore form
an equal-norm tight frame in $\mathbb R^2$. Define the symmetric and
antisymmetric local two-copy modes
\begin{equation}
    \ket{u_r}=\ket0\ket{x_r}+\ket{x_r}\ket0,
    \qquad
    \ket{v_r}=\ket0\ket{x_r}-\ket{x_r}\ket0,
\end{equation}
and
\begin{equation}
    \ket{\Psi_d}
    =
    \frac1{\sqrt2}
    \sum_{r=1}^{2}
    \ket{u_r}_{A_1A_2}\ket{v_r}_{B_1B_2}.
    \label{eq:general-witness}
\end{equation}
Equation~\eqref{eq:frame-identities} makes each pair
$\{\ket{u_r}\}$ and $\{\ket{v_r}\}$ orthonormal, so
$\ket{\Psi_d}$ is normalized and has Schmidt rank two across
$(A_1A_2):(B_1B_2)$. If
$p_j=\sum_r|\braket{j|x_r}|^2$, then $\sum_jp_j=1$ and, within this
ansatz,
\begin{equation}
    \bra{\Psi_d}\Delta\otimes\Delta\ket{\Psi_d}
    =\sum_jp_j^2\geq\frac1{d-1},
\end{equation}
with equality exactly for equal-norm spreading. Thus the frame
minimizes the unavoidable positive diagonal contribution.

A direct dimension-independent evaluation gives
\begin{align}
    \bra{\Psi_d}\Delta\otimes\Delta\ket{\Psi_d}
    &=\frac1{d-1},
    \label{eq:expval-DD}\\
    \bra{\Psi_d}\Delta\otimes(I-\Omega)\ket{\Psi_d}
    &=-\frac14,
    \label{eq:expval-DJ}\\
    \bra{\Psi_d}(I-\Omega)\otimes\Delta\ket{\Psi_d}
    &=-\frac14,
    \label{eq:expval-JD}\\
    \bra{\Psi_d}(I-\Omega)^{\otimes2}\ket{\Psi_d}
    &=-\frac12.
    \label{eq:expval-JJ}
\end{align}
The derivation is given in the Appendix. Hence
\begin{align}
    \bra{\Psi_d}Q_d^{\otimes2}\ket{\Psi_d}
    &=
    \frac1{d-1}-\frac1d-\frac2{d^2}
    \nonumber\\
    &=-\frac{d-2}{d^2(d-1)},
\end{align}
and therefore
\begin{equation}
    \boxed{
    \bra{\Psi_d}
    \left(\rho_{C,d}^{\Gamma}\right)^{\otimes2}
    \ket{\Psi_d}
    =
    -\frac{d-2}{d^2(d-1)(3d-2)^2}<0.
    }
    \label{eq:general-negative}
\end{equation}
This proves Theorem~\ref{thm:pointC}. For $d=3$ the value is
$-1/882$; the corresponding sparse vector and its collective-filtering
realization are given in the Appendix.

\subsection{An open two-copy-distillable region}

The same certificate also provides an analytic inner bound on the
true two-copy-distillable set around $C$. Define
\begin{equation}
    w_{C,d}(b,c)
    =
    \bra{\Psi_d}
    \left(\rho_{b,c}^{\Gamma}\right)^{\otimes2}
    \ket{\Psi_d},
    \label{eq:two-copy-region-definition}
\end{equation}
which is the explicit quadratic polynomial given in
Eq.~\eqref{eq:fixed-witness-plane-general}. Within the physical
canonical plane, the connected negative component containing $C$ is a
relatively open hyperbolic lobe for every $d\geq3$. Parameterizing the
edges $CB$ and $CG$ as
\begin{equation}
    \rho_{CX}(t)=(1-t)\rho_{C,d}+t\rho_{X,d},
    \qquad X=B,G,
\end{equation}
their first boundary intersections occur at
\begin{equation}
\begin{aligned}
    t_B^{(d)}
    &=
    \frac{8(d-1)}
    {3(d-1)(d+2)+(3d-2)\sqrt{(d-1)(d+7)}},\\
    t_G^{(d)}
    &=
    \frac{(d-2)(2d-1)}{5d^2-10d+4}.
\end{aligned}
\label{eq:main-lobe-cutoffs}
\end{equation}
Thus the witness certifies finite segments of both $CB$ and $CG$ in
all dimensions. For $d=3$ these fractions are approximately $0.261$
and $0.263$, producing the blue lobe in Fig.~\ref{fig:qutrit-geometry}.
The full quadratic and its geometry are derived in
the Appendix; failure of this fixed witness beyond the lobe does not
imply undistillability.

Although this lobe is drawn in the two-dimensional canonical plane,
the separation is not confined to that slice. At $B$,
\begin{equation}
    \rho_{B,d}^{\Gamma}
    =\frac{I+(d-1)\Delta-\Omega}{2d(d-1)}\succeq0,
    \qquad
    \ker\rho_{B,d}^{\Gamma}=\operatorname{span}\{\Phi_d\}.
    \label{eq:B-strict-margin}
\end{equation}
Since $\SR(\Phi_d)=d>2$, compactness implies that the minimum of
$\bra{\psi}\rho_{B,d}^{\Gamma}\ket{\psi}$ over normalized
Schmidt-rank-at-most-two vectors is strictly positive. Choose any point
in the interior of both the lobe and $BCG$. Its convex decomposition
has a positive $B$ weight, so it inherits a positive one-copy margin
while the fixed two-copy expectation is negative. Both strict
inequalities persist under sufficiently small arbitrary perturbations.
Hence one-copy-undistillable but two-copy-distillable states contain an
open set in the full state space.

The complementary edge $CG$ has an exact zero-margin structure. For
every $d\geq3$ and every $\rho\in CG$, there is a normalized
Schmidt-rank-two vector $\ket{\zeta_d}$ with
\begin{equation}
    \bra{\zeta_d}\rho^\Gamma\ket{\zeta_d}=0.
    \label{eq:CG-zero-main}
\end{equation}
Tensoring it with product vectors gives a zero Schmidt-rank-two
direction at every copy number. Thus any $k$-copy-undistillable point
on $CG$ has exactly zero minimum in Eq.~\eqref{eq:k-copy-criterion}.
The construction is given in the Appendix.

\subsection{Three-copy fixed witnesses}
\label{sec:three-copy-fixed}

We next construct three-copy Schmidt-rank-two witnesses whose fixed
negative regions extend beyond the two-copy lobe. We use a sparse
qutrit witness $\ket{\Xi_3}$ and, for $d\geq4$, a three-line frame
witness $\ket{\Theta_{d,3}}$. Their explicit forms and negative
point-$C$ expectations are given in the Appendix. Their fixed cubic
expectations certify states beyond the two-copy lobe. Along $CB$, for example, 
the certified fraction increases from $0.261$ to $0.497$ for $d=3$, and 
from $0.215$ to $0.296$ for $d=4$; along $CG$, the fixed two-copy certificate 
remains at least as strong for $d=3,\ldots,6$. The exact cubics, cutoff table,
and a schematic of the certified regions are also given in
the Appendix. These improvements concern the
available certificates only: they do not by themselves establish a
strict separation between the optimized two- and three-copy
boundaries.

\section{Discussion}

The canonical one-copy-undistillable region contains states with
opposite two-copy behavior. The recent Werner results and the
propagation lemma place $BGK$ outside the two-copy-distillable set.
Beyond this, the exact decomposition criterion in the Appendix
certifies an additional two-copy-undistillable band extending from the
$BG$ side into $BCG$. For $d=3$, exact rational certificates cover a
polygonal subset of $BCG$ with area fraction $107/300\simeq35.7\%$;
the same construction certifies slightly larger subsets for each
$d=4,\ldots,8$. From the other side, the lobe around $C$ lies inside
the two-copy-distillable set and implies an open set in the full state
space of one-copy-undistillable but two-copy-distillable states. An
intermediate region remains unresolved.

Numerical searches reported in the Appendix suggest that for small
$d$ the fixed-witness lobe is essentially the full two-copy-distillable
part of $BCG$, whereas for $d=7,8$ other Schmidt-rank-two directions
certify states slightly beyond it.

At the all-copy level, Ref.~\cite{DiVincenzo2000} showed that the
canonical problem can be reduced further to Werner states. At finite
copy number, however, this reduction loses information. The
$U\otimes U$ twirl sends $C$ to
$\alpha_C=-(d+2)/(3d+2)$, strictly inside the
two-copy-undistillable Werner segment $GH$, even though $C$ itself is
two-copy distillable. Thus Werner twirling need not preserve the
number of copies required for distillation. Since the canonical
filtering-and-twirling procedure is local preprocessing, any input
state whose canonical image satisfies $w_{C,d}(b,c)<0$ is also
two-copy distillable.

For each local dimension $d$, let $\mathcal D_k^{(d)}$ denote the
subset of the canonical NPT region distillable using at most $k$
copies. Then
\begin{equation}
    \mathcal D_1^{(d)}\subseteq\mathcal D_2^{(d)}
    \subseteq\mathcal D_3^{(d)}\subseteq\cdots .
    \label{eq:Dk-hierarchy}
\end{equation}
If $\mathcal N_{\mathrm{1-und}}^{(d)}$ is the NPT part of the canonical
one-copy-undistillable region, define
\begin{equation}
    \mathcal U_\infty^{(d)}
    =\mathcal N_{\mathrm{1-und}}^{(d)}
    \setminus\bigcup_{k\geq1}\mathcal D_k^{(d)} .
    \label{eq:Uinfty}
\end{equation}
A nonempty $\mathcal U_\infty^{(d)}$ would give NPT bound
entanglement. For this all-copy existence question, it is enough to
consider Werner states: a further Werner twirl maps an undistillable
canonical NPT state to an undistillable NPT Werner state, while Werner
states already belong to the canonical family. Conversely, if
$\mathcal U_\infty^{(d)}$ is empty for every $d\geq3$, the canonical
reduction of Ref.~\cite{DiVincenzo2000} implies finite-copy
distillability for every bipartite NPT state.

The edges $BC$ and $CG$ give complementary one-parameter views of this
finite-copy structure. Along $BC$,
\begin{equation}
    \rho_p=(1-p)\frac{\Delta}{d}
    +p\frac{P_-}{\binom d2},
    \qquad P_-=\frac{I-\swap}{2},
    \label{eq:BC-benchmark}
\end{equation}
with $B$ at $p=1/2$ and $C$ at $p=2(d-1)/(3d-2)$. For $d=3$ the fixed
distillability certificates reach $0.261$ of the interval from $C$
toward $B$ at two copies and $0.497$ at three, while the exact two-copy
decomposition certificates cover a substantial complementary portion
adjacent to $B$. Along $CG$, by contrast, any $k$-copy-undistillable
point has minimum exactly zero in Eq.~\eqref{eq:k-copy-criterion}. For
qutrits the fixed two- and three-copy witnesses share the exact cutoff
$t=5/19$. Whether the remaining segment to $G$ is undistillable remains
open. Because a zero Schmidt-rank-two direction is always present on
$CG$, a proof of undistillability there cannot rely on a strictly
positive lower bound.

The parity-and-frame construction and the higher-copy witnesses give
explicit tools for studying these finite-copy boundaries. The
higher-copy status of the Werner threshold $G$ remains open; recent
work proves three-copy undistillability for broad classes of test
vectors, but not in full generality~\cite{WuZou2026}. Determining the
finite-copy boundaries in the canonical plane would quantify how many
copies are needed and which distinctions are lost under full twirling.
The all-copy existence question, however, remains equivalent to the
Werner problem.

\vspace{0.3cm}
\begin{acknowledgments}
KSC acknowledges support from the National Science and Technology Council,
Taiwan (Grants No. 109-2112-M-006-010-MY3, 112-2628-M-006-007-MY4).

\vspace{0.2cm}
\noindent\textit{AI-use disclosure.---}
GNMT acknowledges substantial assistance from OpenAI's ChatGPT, using
the GPT-5.6 Thinking model, in the discovery and development of this
work, including scientific reasoning, numerical searches, derivation of
exact two-copy and three-copy certificates, and organization of the
arguments. The interaction was initiated and directed by GNMT.
Anthropic's Claude (Claude Opus 5.5) was used during manuscript
preparation for independent numerical and symbolic cross-checks,
verification of certificate counts and figure geometry, and editorial
review; it did not contribute to the discovery of the results. The
authors independently reconstructed and verified all calculations and
assume full responsibility for the correctness of the results. A more
detailed account of the AI assistance is given in the Appendix.
\end{acknowledgments}

\vspace{0.2cm}
\noindent\textit{Data availability.---}
The exact rational certificate data and standalone verifier supporting
the exact certificate construction in the Appendix are provided as
ancillary files with the submission.
The numerical scripts and data underlying the exploratory searches
reported in the Appendix are available from the corresponding author
upon reasonable request.

\bibliography{Paper}

\clearpage
\appendix
\onecolumngrid

\section{Coefficient-matrix motivation and equal-norm-frame
expectation values}
\label{app:frame-expval}

\subsection{Coefficient-matrix form of the point-\texorpdfstring{$C$}{C} terms}

Let
\begin{equation}
    \ket{\Psi}=\operatorname{vec}_{\operatorname{r}}(M)
    :=
    \sum_{a,j,b,k=0}^{d-1}
    M_{(a,j),(b,k)}
    \ket{a}_{A_1}\ket{j}_{A_2}
    \ket{b}_{B_1}\ket{k}_{B_2},
    \label{eq:general-M-vectorization}
\end{equation}
where $M$ is the $d^2\times d^2$ coefficient matrix of
$\ket{\Psi}$ across $(A_1A_2):(B_1B_2)$. With this convention, the
Schmidt rank of $\ket{\Psi}$ equals the matrix rank of $M$, and
normalization is equivalent to $\|M\|_{\mathrm F}=1$.

Write $M$ as a $d\times d$ array of $d\times d$ blocks,
\begin{equation}
    M=(M_{ab})_{a,b=0}^{d-1},
    \qquad
    (M_{ab})_{jk}
    =
    M_{(a,j),(b,k)}.
\end{equation}
Thus the block indices $a,b$ refer to the first copy, while the
indices $j,k$ inside each block refer to the second copy.

The actions of $\Delta$ and $\Omega$ have simple interpretations.
The projector
\begin{equation}
    \Delta=\sum_i\ket{ii}\bra{ii}
\end{equation}
selects equal Alice and Bob indices and adds the corresponding
squared amplitudes. On the other hand,
\begin{equation}
    \Omega=\sum_{i,j}\ket{ii}\bra{jj}
    =
    \ket{s}\bra{s},
    \qquad
    \ket{s}=\sum_i\ket{ii},
\end{equation}
adds the selected amplitudes coherently before taking their squared
magnitude. In what follows, the first tensor factor acts on $A_1B_1$
and the second on $A_2B_2$.

Expanding
\begin{equation}
    Q_d^{\otimes2}
    =
    \Delta\otimes\Delta+
    \frac{2}{d}
    \left[
        \Delta\otimes(I-\Omega)
        +(I-\Omega)\otimes\Delta
    \right]+
    \frac{4}{d^2}
    (I-\Omega)\otimes(I-\Omega),
\label{eq:general-M-Q-expansion}
\end{equation}
we evaluate the four terms separately.

First, $\Delta\otimes\Delta$ imposes $a=b=i$ on the first copy
and $j=k$ on the second copy. Hence
\begin{equation}
    \bra{\Psi}\Delta\otimes\Delta\ket{\Psi}
    =
    \sum_{i,j}
    \left|(M_{ii})_{jj}\right|^2.
    \label{eq:general-M-DD}
\end{equation}

For $\Delta\otimes I$, the condition $a=b=i$ selects the diagonal
block $M_{ii}$, giving $\sum_i\|M_{ii}\|_{\mathrm F}^2$. Replacing
$I$ by $\Omega$ coherently sums the diagonal entries of this block,
giving $\sum_i|\operatorname{tr}M_{ii}|^2$. Therefore,
\begin{equation}
    \bra{\Psi}\Delta\otimes(I-\Omega)\ket{\Psi}
    =
    \sum_i
    \left(
        \|M_{ii}\|_{\mathrm F}^2
        -
        |\operatorname{tr}M_{ii}|^2
    \right).
    \label{eq:general-M-DJ}
\end{equation}

With the copies reversed, $I\otimes\Delta$ imposes $j=k$ and gives
$\sum_j\sum_{a,b}|(M_{ab})_{jj}|^2$. The operator
$\Omega\otimes\Delta$ then coherently sums the remaining diagonal
first-copy indices, giving
$\sum_j|\sum_i(M_{ii})_{jj}|^2$. Thus
\begin{equation}
    \bra{\Psi}(I-\Omega)\otimes\Delta\ket{\Psi}=
    \sum_j
    \left[
        \sum_{a,b}|(M_{ab})_{jj}|^2
        -
        \left|
            \sum_i(M_{ii})_{jj}
        \right|^2
    \right].
\label{eq:general-M-JD}
\end{equation}

Finally, expand
\begin{equation}
    (I-\Omega)\otimes(I-\Omega)
    =
    I\otimes I-\Omega\otimes I
    -I\otimes\Omega+\Omega\otimes\Omega.
\end{equation}
The four relevant expectation values are
\begin{equation}
\begin{aligned}
    \bra{\Psi}I\otimes I\ket{\Psi}
    &=
    \|M\|_{\mathrm F}^2,\\
    \bra{\Psi}\Omega\otimes I\ket{\Psi}
    &=
    \left\|\sum_iM_{ii}\right\|_{\mathrm F}^2,\\
    \bra{\Psi}I\otimes\Omega\ket{\Psi}
    &=
    \sum_{a,b}|\operatorname{tr}M_{ab}|^2,\\
    \bra{\Psi}\Omega\otimes\Omega\ket{\Psi}
    &=
    |\operatorname{tr}M|^2.
\end{aligned}
\end{equation}
Here $\Omega\otimes I$ coherently sums the diagonal blocks, while
$I\otimes\Omega$ takes the trace of each block.
When $\Omega$ acts on both copies, all diagonal coefficients are
summed, producing $\operatorname{tr}M$.
Consequently,
\begin{equation}
    \bra{\Psi}
    (I-\Omega)\otimes(I-\Omega)
    \ket{\Psi}=
    \|M\|_{\mathrm F}^2
    -
    \left\|\sum_iM_{ii}\right\|_{\mathrm F}^2
    -
    \sum_{a,b}|\operatorname{tr}M_{ab}|^2
    +
    |\operatorname{tr}M|^2.
\label{eq:general-M-JJ}
\end{equation}

These expressions suggest seeking a low-rank coefficient matrix for
which the coherent sums are large while the total trace vanishes.
A simple way to impose the latter while treating the copies symmetrically is
\begin{equation}
    SMS^T=-M,
\end{equation}
where $S\ket{a}\ket{j}=\ket{j}\ket{a}$ interchanges the two local
copies. Indeed, since $S^TS=I$,
\begin{equation}
    \operatorname{tr}M
    =
    \operatorname{tr}(SMS^T)
    =
    -\operatorname{tr}M,
\end{equation}
and therefore $\operatorname{tr}M=0$.

To maintain Schmidt rank at most two, take
\begin{equation}
    M=\frac{1}{\sqrt2}\sum_{r=1}^{2}u_rv_r^T.
\end{equation}
A termwise realization of the required swap parity is
\begin{equation}
    Su_r=u_r,
    \qquad
    Sv_r=-v_r.
\end{equation}
Choosing a distinguished basis vector $\ket0$ and vectors
$\ket{x_r}\perp\ket0$, the simplest such local modes are
\begin{equation}
    \ket{u_r}
    =
    \ket0\ket{x_r}+\ket{x_r}\ket0,
    \qquad
    \ket{v_r}
    =
    \ket0\ket{x_r}-\ket{x_r}\ket0.
\end{equation}
This recovers the ansatz of Eq.~\eqref{eq:general-witness}. The
remaining task is to choose the vectors $\ket{x_r}$ so that the
positive diagonal term in Eq.~\eqref{eq:general-M-DD} is as small as
possible.

\subsection{Evaluation for the equal-norm frame}

Put $m=d-1$ and define the Gram operator
\begin{equation}
    K=\sum_{r=1}^{2}\ket{x_r}\!\bra{x_r}
\end{equation}
on $\operatorname{span}\{\ket1,\ldots,\ket m\}$. Extend its matrix by
setting $K_{0j}=K_{j0}=0$. The frame identities
\eqref{eq:frame-identities} imply
\begin{equation}
    K_{jj}=\frac1m,
    \qquad
    \operatorname{tr}K=1,
    \qquad
    \operatorname{tr}K^2=\frac12.
    \label{eq:K-identities}
\end{equation}

Write the coefficients of $\ket{\Psi_d}$ in the ordering
$(\alpha,\gamma;\beta,\delta)=(A_1,A_2;B_1,B_2)$. Expanding
Eq.~\eqref{eq:general-witness} gives
\begin{align}
    \Psi_{\alpha\gamma;\beta\delta}
    &=
    \frac1{\sqrt2}\Big(
    \delta_{\alpha0}\delta_{\beta0}K_{\gamma\delta}
    -\delta_{\alpha0}\delta_{\delta0}K_{\gamma\beta}
    +\delta_{\gamma0}\delta_{\beta0}K_{\alpha\delta}
    -\delta_{\gamma0}\delta_{\delta0}K_{\alpha\beta}
    \Big).
    \label{eq:Psi-coefficients}
\end{align}
This coefficient formula reduces all the required expectation values to
Eq.~\eqref{eq:K-identities}.

First, $\Delta\otimes\Delta$ selects
$\alpha=\beta$ and $\gamma=\delta$. The only nonzero contributions
have one of $\alpha,\gamma$ equal to zero, and therefore
\begin{equation}
    \bra{\Psi_d}\Delta\otimes\Delta\ket{\Psi_d}
    =
    \sum_{j=1}^{m}K_{jj}^2
    =
    \frac1m.
\end{equation}

To evaluate $\bra{\Psi_d}\Omega\otimes I\ket{\Psi_d}$, define
\begin{equation}
    S_{\gamma\delta}
    =
    \sum_{\alpha=0}^{m}
    \Psi_{\alpha\gamma;\alpha\delta}.
\end{equation}
Equation~\eqref{eq:Psi-coefficients} gives
\begin{equation}
    S_{\gamma\delta}
    =
    \frac1{\sqrt2}
    \left(
        K_{\gamma\delta}
        -
        \delta_{\gamma0}\delta_{\delta0}\operatorname{tr}K
    \right).
\end{equation}
Hence
\begin{align}
    \bra{\Psi_d}\Omega\otimes I\ket{\Psi_d}
    &=
    \sum_{\gamma,\delta}|S_{\gamma\delta}|^2
    =
    \frac12
    \left[
        \operatorname{tr}K^2+
        (\operatorname{tr}K)^2
    \right]
    =
    \frac34.
    \label{eq:Omega-I}
\end{align}
The same value holds for $I\otimes\Omega$. Moreover,
\begin{equation}
    \bra{\Psi_d}\Omega\otimes\Omega\ket{\Psi_d}
    =
    \left|\sum_{\gamma}S_{\gamma\gamma}\right|^2
    =0,
    \label{eq:Omega-Omega}
\end{equation}
because the $K$ contribution cancels the $(0,0)$ contribution.

The expectation value of $\Delta\otimes I$ is
\begin{align}
    \bra{\Psi_d}\Delta\otimes I\ket{\Psi_d}
    &=
    \sum_{\alpha,\gamma,\delta}
    \left|\Psi_{\alpha\gamma;\alpha\delta}\right|^2
    =
    \frac12\operatorname{tr}K^2
    +
    \frac12\sum_{j=1}^{m}K_{jj}^2
    =
    \frac14+\frac1{2m}.
    \label{eq:Delta-I}
\end{align}
Similarly,
\begin{align}
    \bra{\Psi_d}\Delta\otimes\Omega\ket{\Psi_d}
    &=
    \sum_{\alpha}
    \left|
        \sum_{\gamma}
        \Psi_{\alpha\gamma;\alpha\gamma}
    \right|^2
    =
    \frac12(\operatorname{tr}K)^2
    +
    \frac12\sum_{j=1}^{m}K_{jj}^2
    =
    \frac12+\frac1{2m}.
    \label{eq:Delta-Omega}
\end{align}
Subtracting Eq.~\eqref{eq:Delta-Omega} from
Eq.~\eqref{eq:Delta-I} yields
\begin{equation}
    \bra{\Psi_d}
    \Delta\otimes(I-\Omega)
    \ket{\Psi_d}
    =
    -\frac14.
\end{equation}
Exchange symmetry gives the corresponding expectation value with the
two copies reversed.
Finally, using normalization together with
Eqs.~\eqref{eq:Omega-I} and~\eqref{eq:Omega-Omega},
\begin{align}
    \bra{\Psi_d}
    (I-\Omega)\otimes(I-\Omega)
    \ket{\Psi_d}
    &=
    1-\frac34-\frac34+0
    =
    -\frac12.
\end{align}
This proves Eqs.~\eqref{eq:expval-DD}--\eqref{eq:expval-JJ}.

\section{Filtering and hidden teleportation interpretation}
\label{app:hidden}

The filtering formulation below provides the operational route by
which the point-$C$ certificate was found. It was motivated by the
notion of hidden teleportation power under local
filtering~\cite{HiddenTeleportationPower}. Clarifying its equivalence 
with $k$-copy distillability when collective local filtering is 
allowed subsequently led to a search for such filters at point $C$.

For a state $\tau$ on $\mathbb C^D\otimes\mathbb C^D$, define its
fully entangled fraction by
\begin{equation}
    F_D(\tau)
    =
    \max_{\ket{\Phi_D}}
    \bra{\Phi_D}\tau\ket{\Phi_D},
\end{equation}
where the maximum is over maximally entangled states of local
dimension $D$. The state is useful for standard teleportation exactly
when~\cite{HorodeckiTeleportation}
\begin{equation}
    F_D(\tau)>\frac1D.
\end{equation}

The relation needed here is a standard consequence of the reduction
criterion and two-qubit filtering. A state $\rho$ is $k$-copy
distillable if and only if collective local filtering on $\rho^{\otimes
k}$ can produce, with nonzero probability, a state $\sigma$ on some
$\mathbb C^D\otimes\mathbb C^D$ satisfying
\begin{equation}
    F_D(\sigma)>\frac1D .
\end{equation}
Moreover, $D=2$ always suffices. Indeed, $k$-copy distillability yields
an entangled two-qubit output, which can be filtered to
$F_2>1/2$~\cite{VerstraeteVerschelde}. Conversely, if
$F_D(\sigma)>1/D$, then for a maximizing maximally entangled state
$\ket{\Phi_D}$,
\begin{equation}
    \bra{\Phi_D}
    \left(\sigma_A\otimes I-\sigma\right)
    \ket{\Phi_D}<0,
\end{equation}
so $\sigma$ violates the reduction criterion. The corresponding
reduction-criterion argument yields an entangled local
$2\otimes2$ projection of a single copy of $\sigma$, and hence
$\sigma$ is one-copy distillable~\cite{Horodecki1999_redcrit}.
Composing this projection with the original filters on
$\rho^{\otimes k}$ proves $k$-copy distillability of $\rho$.
Thus the output dimension is irrelevant to the existence of a
distilling filter.

The dimension $D=d^k$ enters only when we use the hidden-teleportation
definition that compares the filtered output with the teleportation
threshold in that fixed local dimension. For a $d\times d$ state
$\rho$, if in addition
\begin{equation}
    F_{d^k}\!\left(\rho^{\otimes k}\right)\leq\frac1{d^k},
\end{equation}
then filtered teleportation usefulness is genuinely hidden, and
\begin{equation}
    \text{$k$-copy hidden teleportation power}
    \iff
    \text{$k$-copy distillability}.
\end{equation}

\subsection{Qutrit specialization of the uniform witness}

For $d=3$, the equal-norm frame consists of the two orthogonal
coordinate directions
\begin{equation}
    \ket{x_1}=\frac{\ket1}{\sqrt2},
    \qquad
    \ket{x_2}=\frac{\ket2}{\sqrt2}.
\end{equation}
Thus
\begin{align}
    \ket{u_1}&=\frac{\ket{01}+\ket{10}}{\sqrt2},
    &
    \ket{u_2}&=\frac{\ket{02}+\ket{20}}{\sqrt2},
    \\
    \ket{v_1}&=\frac{\ket{01}-\ket{10}}{\sqrt2},
    &
    \ket{v_2}&=\frac{\ket{02}-\ket{20}}{\sqrt2},
\end{align}
and the uniform witness becomes
\begin{equation}
    \ket{\Psi_3}
    =
    \frac1{\sqrt2}
    \left(
        \ket{u_1}_A\ket{v_1}_B
        +
        \ket{u_2}_A\ket{v_2}_B
    \right).
    \label{eq:qutrit-witness}
\end{equation}
Equivalently,
\begin{align}
    \ket{\Psi_3}
    =\frac{1}{2\sqrt2}\Big[{}
    &(\ket{01}+\ket{10})_A
      \otimes(\ket{01}-\ket{10})_B
    \nonumber\\
    &+
    (\ket{02}+\ket{20})_A
      \otimes(\ket{02}-\ket{20})_B
    \Big].
\end{align}
The two product terms have mutually orthogonal local factors, so
$\ket{\Psi_3}$ is normalized and has Schmidt rank exactly two. The
all-dimensional expectation formula gives
\begin{equation}
    \bra{\Psi_3}
    \left(\rho_{C,3}^{\Gamma}\right)^{\otimes2}
    \ket{\Psi_3}
    =-\frac1{882}<0.
    \label{eq:qutrit-negative}
\end{equation}
The filters below implement precisely these symmetric and
antisymmetric frame modes.

\subsection{Exact qutrit filters}

On each local two-qutrit space, use the ordered basis
\begin{equation}
\{
\ket{00},\ket{01},\ket{02},
\ket{10},\ket{11},\ket{12},
\ket{20},\ket{21},\ket{22}
\}.
\end{equation}
Define $A_3,B_3:\mathbb C^9\to\mathbb C^2$ by
\begin{equation}
    A_3
    =
    \frac1{\sqrt2}
    \begin{pmatrix}
        0&1&0&1&0&0&0&0&0\\
        0&0&1&0&0&0&1&0&0
    \end{pmatrix},
    \label{eq:filterA3}
\end{equation}
and
\begin{equation}
    B_3
    =
    \frac1{\sqrt2}
    \begin{pmatrix}
        0&0&-1&0&0&0&1&0&0\\
        0&1& 0&-1&0&0&0&0&0
    \end{pmatrix}.
    \label{eq:filterB3}
\end{equation}
The rows of $A_3$ are $\bra{u_1}$ and $\bra{u_2}$, while the
rows of $B_3$ are $-\bra{v_2}$ and $\bra{v_1}$. In particular,
their rows are orthonormal, so they define valid stochastic local
filters.

The unnormalized output from two copies of $\rho_{C,3}$ is
\begin{equation}
    \omega_3
    =
    \frac1{882}
    \begin{pmatrix}
        10&0&0&10\\
        0&9&0&0\\
        0&0&9&0\\
        10&0&0&10
    \end{pmatrix},
\end{equation}
with success probability
\begin{equation}
    p_3
    =
    \operatorname{tr}\omega_3
    =
    \frac{19}{441}.
\end{equation}
The normalized state is
\begin{equation}
    \sigma_3
    =
    \frac1{38}
    \begin{pmatrix}
        10&0&0&10\\
        0&9&0&0\\
        0&0&9&0\\
        10&0&0&10
    \end{pmatrix}.
\end{equation}
In the Bell basis,
\begin{equation}
    \sigma_3
    =
    \frac{10}{19}\ket{\Phi^+}\!\bra{\Phi^+}
    +
    \frac{9}{38}\ket{\Psi^+}\!\bra{\Psi^+}
    +
    \frac{9}{38}\ket{\Psi^-}\!\bra{\Psi^-}.
\end{equation}
Therefore
\begin{equation}
    F_2(\sigma_3)=\frac{10}{19}>\frac12.
\end{equation}

The negative eigenvalue of $\sigma_3^\Gamma$ is $-1/38$, with
eigenvector
\begin{equation}
    \ket{\eta}
    =
    \frac{\ket{01}-\ket{10}}{\sqrt2}.
\end{equation}
Indeed,
\begin{equation}
    A_3^\dagger\ket0=\ket{u_1},
    \qquad
    A_3^\dagger\ket1=\ket{u_2},
\end{equation}
while
\begin{equation}
    B_3^T\ket0=-\ket{v_2},
    \qquad
    B_3^T\ket1=\ket{v_1}.
\end{equation}
It follows that
\begin{align}
    (A_3^\dagger\otimes B_3^T)\ket{\eta}
    &=
    \frac1{\sqrt2}
    \left(
        \ket{u_1}\ket{v_1}
        +
        \ket{u_2}\ket{v_2}
    \right)
    \nonumber\\
    &=
    \ket{\Psi_3}.
\end{align}
Since
\begin{equation}
    \omega_3^\Gamma
    =
    (A_3\otimes\overline B_3)
    \left(\rho_{C,3}^{\Gamma}\right)^{\otimes2}
    (A_3^\dagger\otimes B_3^T),
\end{equation}
the lifted expectation is
\begin{equation}
    \bra{\Psi_3}
    \left(\rho_{C,3}^{\Gamma}\right)^{\otimes2}
    \ket{\Psi_3}
    =
    p_3\left(-\frac1{38}\right)
    =
    -\frac1{882}.
\end{equation}

\subsection{Hidden teleportation power of point \texorpdfstring{$C$}{C}}

The nonzero eigenvalues of $\rho_{C,d}$ are $1/(3d-2)$ on the
diagonal subspace and $4/[d(3d-2)]$ on the antisymmetric subspace.
Therefore
\begin{equation}
    \left\|\rho_{C,d}\right\|_\infty
    =
    \frac1{3d-2}\max\!\left\{1,\frac4d\right\}
    <
    \frac1d
    \qquad(d\geq3).
\end{equation}
It follows that
\begin{align}
    F_d(\rho_{C,d})
    &<
    \frac1d,
    \\
    F_{d^2}(\rho_{C,d}^{\otimes2})
    &<
    \frac1{d^2}.
\end{align}
One-copy undistillability and the standard filtering equivalence above
rule out every one-copy filter
that produces a teleportation-useful output. By
Theorem~\ref{thm:pointC}, two copies can be filtered into a useful
two-qubit state. To match the fixed-output-dimension convention used
in the hidden-teleportation definition, this output can be embedded
into local dimension $d^2$ with
\begin{equation}
    F_{d^2}(\widehat{\sigma})
    \geq
    \frac2{d^2}F_2(\sigma)
    >
    \frac1{d^2}.
\end{equation}
Thus point $C$ has no one-copy hidden teleportation power but has
two-copy hidden teleportation power for every $d\geq3$.

\subsection*{AI-assisted discovery and independent verification}

The interaction that led to the present construction began with an
attempt by GNMT to clarify the equivalence between $k$-copy
distillability and $k$-copy teleportation usefulness under collective
local filtering. OpenAI's ChatGPT, using the GPT-5.6 Thinking model,
assisted in establishing the relevant filtering argument. This led to
the proposal of a numerical search for collective local filters at
point $C$ of the canonical NPT family. ChatGPT assisted in formulating
and implementing the search, which identified qutrit filters producing
a teleportation-useful two-qubit state, and subsequently in reducing
the numerical solution to an exact qutrit certificate.

GNMT then suggested choosing the qutrit solution in a form that could
be incorporated into a common construction for all $d\geq3$, leading
to the uniform equal-norm tight-frame formulation. ChatGPT further
assisted in deriving this all-dimensional construction, extending the
fixed witness away from point $C$, and developing the
point-$C$-tailored three-copy certificates. In subsequent work on the
surrounding canonical region, ChatGPT also assisted in developing the
exact decomposition certificates extending two-copy undistillability
from the Werner side into $BCG$, including the finite exact-rational
verification procedure, and in designing and analyzing the exploratory
numerical searches reported in Appendix~\ref{app:numerics}. These numerical
investigations were used to test the sharpness and limitations of the
analytic certificates rather than as proofs of undistillability.

Separately, during manuscript preparation, Anthropic's Claude
(Claude Opus 5.5) was used as an independent checker and reviewer. It
re-evaluated the explicit formulas numerically and symbolically,
including the frame expectation values, the fixed-witness polynomial
and its cutoffs, the qutrit filters and output state, and the
three-copy witnesses, cubics, and cutoff table. Using a separate 
linear-programming implementation, it reproduced the certified-cell 
counts for $d=3,4,5$, and an independent alternating search reproduced 
the negative $d=7$ value reported in Appendix~\ref{app:numerics}.
It also reconstructed the boundary of the
certified two-copy-undistillable region for $d=3$ used to guide the
schematic green region in Fig.~\ref{fig:qutrit-geometry}, and provided
feedback on presentation and framing. It did not contribute to the
discovery of the results.

The authors independently reconstructed and verified all calculations
and exact certificates used in the manuscript.

\section{The fixed witness \texorpdfstring{$\ket{\Psi_d}$}{Psi d} away from
\texorpdfstring{$C$}{C}}
\label{app:nearby}

For each $d\geq3$, let $\ket{\Psi_d}$ denote the uniform
equal-norm-frame vector in Eq.~\eqref{eq:general-witness}. For $d=3$,
this is precisely the qutrit vector in
Eq.~\eqref{eq:qutrit-witness}.

For an arbitrary canonical state, its partial transpose can be written
as
\begin{equation}
    \rho_{b,c}^{\Gamma}
    =
    cI
    +
    \frac{b-c}{2}(I-\Omega)
    +
    \left[
        \frac1d
        -\frac{d-1}{2}b
        -\frac{d+1}{2}c
    \right]\Delta.
    \label{eq:canonical-PT-general}
\end{equation}
Using the expectation values for $\ket{\Psi_d}$ derived in
Eqs.~\eqref{eq:expval-DD}--\eqref{eq:expval-JJ},
together with
\begin{equation}
    \bra{\Psi_d}I\otimes(I-\Omega)\ket{\Psi_d}
    =
    \frac14
\end{equation}
and
\begin{equation}
    \bra{\Psi_d}I\otimes\Delta\ket{\Psi_d}
    =
    \frac{d+1}{4(d-1)},
\end{equation}
and the corresponding expressions with the two copies exchanged, we
obtain
\begin{equation}
    \boxed{
    \begin{aligned}
    w_{C,d}(b,c)
    &:=
    \bra{\Psi_d}
    \left(\rho_{b,c}^{\Gamma}\right)^{\otimes2}
    \ket{\Psi_d}
    \\[1mm]
    &=
    \frac{1}{8d^2(d-1)}
    \Big[
        d^2(d-1)(3d-4)b^2
        +2d^2(d-1)(d+4)bc
        \\
        &\hspace{18mm}
        -d^2(d-1)(d-4)c^2
        -10d(d-1)b
        -2d(d+3)c
        +8
    \Big].
    \end{aligned}
    }
    \label{eq:fixed-witness-plane-general}
\end{equation}
Every physical canonical state satisfying
\begin{equation}
    w_{C,d}(b,c)<0
\end{equation}
is therefore two-copy distillable.

At point $C$,
\begin{equation}
    w_{C,d}(C)
    =
    -\frac{d-2}
    {d^2(d-1)(3d-2)^2}<0.
\end{equation}
For $d=3$, Eq.~\eqref{eq:fixed-witness-plane-general} reduces to
\begin{equation}
    w_{C,3}(b,c)
    =
    \frac{
        45b^2+126bc-30b+9c^2-18c+4
    }{72},
\end{equation}
recovering the qutrit expression.

\subsection{The segment \texorpdfstring{$BC$}{BC}}

It is useful to view this edge as a one-parameter benchmark,
\begin{equation}
    \rho_p=(1-p)\frac{\Delta}{d}
    +p\frac{P_-}{\binom d2},
    \qquad
    P_-=\frac{I-\swap}{2}.
    \label{eq:BC-edge-p}
\end{equation}
Here $B$ corresponds to $p=1/2$, the state is NPT for $p>1/2$, and
$C$ corresponds to $p=2(d-1)/(3d-2)$. The interval from $B$ to $C$ is
one-copy undistillable, whereas states immediately beyond $C$ along
this edge are one-copy distillable. Moreover, for every fixed copy
number $k$, the fixed-copy neighborhood result of
Ref.~\cite{DiVincenzo2000} leaves a nonzero $k$-copy-undistillable
interval adjacent to $B$.

The relevant points are
\begin{equation}
    B=
    \left(
        \frac1{d(d-1)},0
    \right),
    \qquad
    C=
    \left(
        \frac4{d(3d-2)},0
    \right).
\end{equation}
Parameterize the segment from $C$ to $B$ by
\begin{equation}
    \rho_{CB}(t)
    =
    (1-t)\rho_{C,d}+t\rho_{B,d},
    \qquad
    0\leq t\leq1.
\end{equation}
Then
\begin{equation}
    b(t)
    =
    \frac4{d(3d-2)}
    -
    \frac{d-2}
    {d(d-1)(3d-2)}t,
    \qquad
    c(t)=0.
\end{equation}
Substitution into
Eq.~\eqref{eq:fixed-witness-plane-general} gives
\begin{equation}
    \boxed{
    \begin{aligned}
    w_{CB}^{(d)}(t)
    &=
    \frac{d-2}
    {8d^2(d-1)^2(3d-2)^2}
    \Big[
        (d-2)(3d-4)t^2
        % \\
        % &\hspace{28mm}
        +6(d-1)(d+2)t
        -8(d-1)
    \Big].
    \end{aligned}
    }
    \label{eq:general-BC-witness}
\end{equation}
Its first zero occurs at
\begin{equation}
    \boxed{
    t_B^{(d)}
    =
    \frac{
        8(d-1)
    }{
        3(d-1)(d+2)
        +(3d-2)\sqrt{(d-1)(d+7)}
    }.
    }
    \label{eq:general-tB}
\end{equation}
Hence the fixed witness $\ket{\Psi_d}$ proves two-copy distillability for
\begin{equation}
    0\leq t<t_B^{(d)}.
\end{equation}
For $d=3$,
\begin{equation}
    t_B^{(3)}
    =
    -6+\frac{14\sqrt5}{5}
    \approx0.260990.
\end{equation}

\subsection{The segment \texorpdfstring{$CG$}{CG}}

The point $G$ has coordinates
\begin{equation}
    G=
    \left(
        \frac3{d(2d-1)},
        \frac1{d(2d-1)}
    \right).
\end{equation}
Parameterize
\begin{equation}
    \rho_{CG}(t)
    =
    (1-t)\rho_{C,d}+t\rho_{G,d},
    \qquad
    0\leq t\leq1.
\end{equation}
Then
\begin{equation}
    b(t)
    =
    \frac4{d(3d-2)}
    +
    \frac{d-2}
    {d(2d-1)(3d-2)}t,
\end{equation}
and
\begin{equation}
    c(t)
    =
    \frac{t}{d(2d-1)}.
\end{equation}
The fixed-witness expectation factorizes as
\begin{equation}
    \boxed{
    \begin{aligned}
    w_{CG}^{(d)}(t)
    =
    \frac{
        \bigl[(d-1)t+2d-1\bigr]
        \bigl[
            (5d^2-10d+4)t
            -(d-2)(2d-1)
        \bigr]
    }{
        d^2(d-1)(2d-1)^2(3d-2)^2
    }.
    \end{aligned}
    }
    \label{eq:general-CG-witness}
\end{equation}
The first factor is strictly positive on $0\leq t\leq1$.
Consequently, the fixed witness proves two-copy distillability for
\begin{equation}
    0\leq t<t_G^{(d)},
\end{equation}
where
\begin{equation}
    \boxed{
    t_G^{(d)}
    =
    \frac{
        (d-2)(2d-1)
    }{
        5d^2-10d+4
    }.
    }
    \label{eq:general-tG}
\end{equation}
For $d=3$,
\begin{equation}
    t_G^{(3)}
    =
    \frac5{19}
    \approx0.263158.
\end{equation}

The endpoint $G$ is two-copy undistillable by the Werner-state
result. The present argument does not decide the status of the
remaining part of the segment between $t_G^{(d)}$ and $G$.

\subsection{Geometry of the fixed-witness region}

Define
\begin{equation}
    \mathcal R_d^{(2)}
    =
    \left\{
        (b,c)\ \text{physical}:
        w_{C,d}(b,c)<0
    \right\}.
    \label{eq:fixed-witness-region}
\end{equation}
The quadratic part of the numerator in
Eq.~\eqref{eq:fixed-witness-plane-general} has coefficient matrix
\begin{equation}
    Q_d^{(2)}
    =
    d^2(d-1)
    \begin{pmatrix}
        3d-4 & d+4\\
        d+4  & 4-d
    \end{pmatrix},
\end{equation}
whose determinant is
\begin{equation}
    \det Q_d^{(2)}
    =
    -4d^4(d-1)^2(d^2-2d+8)<0.
\end{equation}
The augmented conic determinant is
\begin{equation}
    -8d^4(d-1)(d^2+6d-9)\neq0,
\end{equation}
so the boundary
\begin{equation}
    w_{C,d}(b,c)=0
\end{equation}
is a nondegenerate hyperbola.

The connected component $\mathcal R_{d,C}^{(2)}$ containing point $C$
is an open hyperbolic lobe. Let
\begin{equation}
    B_d'=\rho_{CB}\left(t_B^{(d)}\right),
    \qquad
    G_d'=\rho_{CG}\left(t_G^{(d)}\right).
\end{equation}
These are the first intersections of its boundary with the segments
$BC$ and $CG$, respectively. Hence the segments from $C$ to $B_d'$
and $G_d'$, excluding their boundary endpoints, are contained in
$\mathcal R_{d,C}^{(2)}$.

Finally, the certified fractions along the edges $CB$ and $CG$
behave asymptotically as
\begin{equation}
    t_B^{(d)}
    =
    \frac{4}{3d}+O(d^{-2}),
    \qquad
    t_G^{(d)}
    =
    \frac25+O(d^{-1}).
    \label{eq:fixed-witness-asymptotic-cutoffs}
\end{equation}
Hence the fixed certificate becomes increasingly elongated in the
direction from $C$ toward $G$ as the local dimension grows.

\subsection{Sector-weight coordinates and the large-\texorpdfstring{$d$}{d} limit}

The raw parameters $b$ and $c$ scale as $d^{-2}$, so comparisons
between dimensions are clearer in terms of the total weights carried by
the three sectors,
\begin{equation}
    w_\Delta=da,
    \qquad
    w_- = \binom d2 b,
    \qquad
    w_+ = \binom d2 c,
    \qquad
    w_\Delta+w_-+w_+=1.
    \label{eq:sector-weights}
\end{equation}
In these coordinates
\begin{align}
    B&=\left(\frac12,\frac12,0\right),
    &
    K&=\left(0,\frac12,\frac12\right),
    \\
    C&=\left(\frac{d}{3d-2},\frac{2(d-1)}{3d-2},0\right),
    &
    G&=\left(\frac1{2d-1},
        \frac{3(d-1)}{2(2d-1)},
        \frac{d-1}{2(2d-1)}\right).
\end{align}
Thus
\begin{equation}
    C\longrightarrow\left(\frac13,\frac23,0\right),
    \qquad
    G\longrightarrow\left(0,\frac34,\frac14\right),
\end{equation}
while $B$ and $K$ are fixed. The triangle \(BCG\) therefore converges
to a nondegenerate limiting triangle.

Substituting Eq.~\eqref{eq:sector-weights} into the fixed-witness
polynomial gives, at fixed sector weights,
\begin{equation}
    8d^2(d-1)w_{C,d}
    =
    -4w_\Delta\bigl(3w_- - w_+ -2\bigr)
    +O(d^{-1}).
    \label{eq:large-d-weight-witness}
\end{equation}
The leading-order zero
$3w_- -w_+-2=0$ is exactly the line through the limiting points $C$
and $G$. Equation~\eqref{eq:fixed-witness-asymptotic-cutoffs} shows that
the certified fraction from $C$ toward $B$ shrinks as $4/(3d)$,
whereas the fraction from $C$ toward $G$ approaches $2/5$. Thus, as
$d$ grows, the region certified by the fixed witness becomes a thin
region along $CG$: its width away from $CG$ scales as $1/d$, while its
extent along $CG$ remains finite. Its fractional area inside $BCG$
therefore tends to zero. This conclusion concerns only the fixed
witness $\ket{\Psi_d}$; optimization over all Schmidt-rank-two vectors
can give a larger region.

\section{Exact two-copy-undistillability certificates inside \texorpdfstring{$BCG$}{BCG}}
\label{app:undist-band}

The recent two-copy Werner theorem~\cite{WernerTwoCopyFu,SongChen2026,
WernerTwoCopyFraser,WernerTwoCopyBharti} can also be used to certify
points inside $BCG$, rather than only the triangle $BGK$.  Let
$X_P=\rho_{P,d}^{\Gamma}$ and introduce the three orthogonal sector
projectors
\begin{equation}
    \Pi_0=\ket{\Phi_d}\!\bra{\Phi_d},\qquad
    \Pi_1=\Delta-\Pi_0,\qquad
    \Pi_2=I-\Delta,
    \label{eq:sector-projectors}
\end{equation}
with dimensions $1$, $d-1$, and $d^2-d$, respectively. Every canonical
partial transpose is diagonal in this decomposition,
\begin{equation}
    X_P=\sum_{s=0}^{2}p_s(P)\Pi_s,
    \label{eq:sector-decomposition}
\end{equation}
and the sector vector $p(P)$ is affine in the canonical coordinates
$(b,c)$.

For exact arithmetic it is convenient to rescale the Werner endpoint
by a positive factor,
\begin{equation}
    W_d:=d(2d-1)X_G
    =(2-d)\Pi_0+2\Pi_1+2\Pi_2.
    \label{eq:scaled-Werner-endpoint}
\end{equation}
We use three types of operators that are nonnegative on every
Schmidt-rank-at-most-two vector: $W_d^{\otimes2}$, by the Werner theorem; 
$Y_i\otimes W_d$ and $W_d\otimes Y_i$, where the $Y_i$ are separable operators 
obtained by twirling product projectors; and positive-semidefinite operators. 
Four such $Y_i$ suffice. In terms of their sector weights
$\bigl(\operatorname{Tr}Y_i\Pi_0,\operatorname{Tr}Y_i\Pi_1,
\operatorname{Tr}Y_i\Pi_2\bigr)$, they are
\begin{equation}
    (0,0,1),\quad
    \left(0,\frac12,\frac12\right),\quad
    \left(\frac1d,1-\frac1d,0\right),\quad
    \left(\frac1d,0,1-\frac1d\right).
    \label{eq:twirled-product-sector-weights}
\end{equation}
For example, these are obtained by the canonical symmetry twirl of the
product projectors onto
$\ket0\ket1$,
$[(\ket0+\ket1)/\sqrt2]\otimes[(\ket0-\ket1)/\sqrt2]$,
$\ket0\otimes\ket0$, and $\ket+\otimes\ket+$, respectively, where
$\ket+=d^{-1/2}\sum_j\ket j$. The mixed terms are valid because
each $Y_i$ is separable: resolving it into product projectors and
contracting the corresponding copy of a global Schmidt-rank-two vector
leaves a vector of Schmidt rank at most two on the other copy, exactly
as in the propagation argument of Ref.~\cite{DiVincenzo2000}. Hence
one-copy undistillability of $G$ applies.

Accordingly, a decomposition
\begin{equation}
\begin{aligned}
    X_P^{\otimes2}
    ={}&\lambda W_d^{\otimes2}
    +\sum_i\left(a_iY_i\otimes W_d+b_iW_d\otimes Y_i\right)+R,\\
    &\lambda,a_i,b_i\geq0,\qquad R\succeq0,
    \label{eq:undist-decomposition}
\end{aligned}
\end{equation}
is a sufficient certificate of two-copy undistillability.

A finite certificate can cover an entire two-dimensional region.  Let
$P_1,P_2,P_3$ be the vertices of a triangle in $BCG$ and
$P=\sum_v\beta_vP_v$ with $\beta_v\geq0$ and $\sum_v\beta_v=1$.
Affineness of Eq.~\eqref{eq:sector-decomposition} gives
\begin{equation}
\begin{aligned}
 X_P^{\otimes2}={}&
 \sum_v\beta_v^2X_{P_v}^{\otimes2}\\
 &+\sum_{v<w}\beta_v\beta_w
 \left(X_{P_v}\otimes X_{P_w}+X_{P_w}\otimes X_{P_v}\right).
 \label{eq:triangle-barycentric-certificate}
\end{aligned}
\end{equation}
Thus it is enough to certify the six fixed vertex-pair targets on the
right-hand side.  In the supplied certificates the residual has the
form
$R=\sum_{s,t}r_{st}\Pi_s\otimes\Pi_t$ with $r_{st}\geq0$, so its
positive semidefiniteness is immediate.

We discretize $BCG$ by taking the convex weights of $B$, $C$, and $G$
in multiples of $1/N$. This produces a regular triangular grid of
$N^2$ equal-area cells. For $N=120$ there are $14400$ cells. Exact
rational certificates of the form above give the following rigorously
certified subsets:
\begin{center}
\begin{tabular}{c|c|c}
$d$ & certified triangles & fraction of $BCG$\\ \hline
3 & 5136 & $107/300$\\
4 & 5152 & $161/450$\\
5 & 5168 & $323/900$\\
6 & 5180 & $259/720$\\
7 & 5183 & $5183/14400$\\
8 & 5193 & $577/1600$
\end{tabular}
\end{center}
On this fixed grid the certified sets satisfy
\begin{equation}
    \mathcal S_3\subseteq\mathcal S_4\subseteq\cdots
    \subseteq\mathcal S_8.
\end{equation}
Thus every cell certified at one of the dimensions $d=3,\ldots,7$ is
also certified at the next dimension. This is an exact statement for
these dimensions on this particular grid; it does not imply
monotonicity for arbitrary $d$. Figure~\ref{fig:exact-band-d3} shows
the $d=3$ set together with the exact negative region of the fixed
witness $\ket{\Psi_d}$.

\begin{figure}[h]
    \centering
    \includegraphics[width=0.5\linewidth]{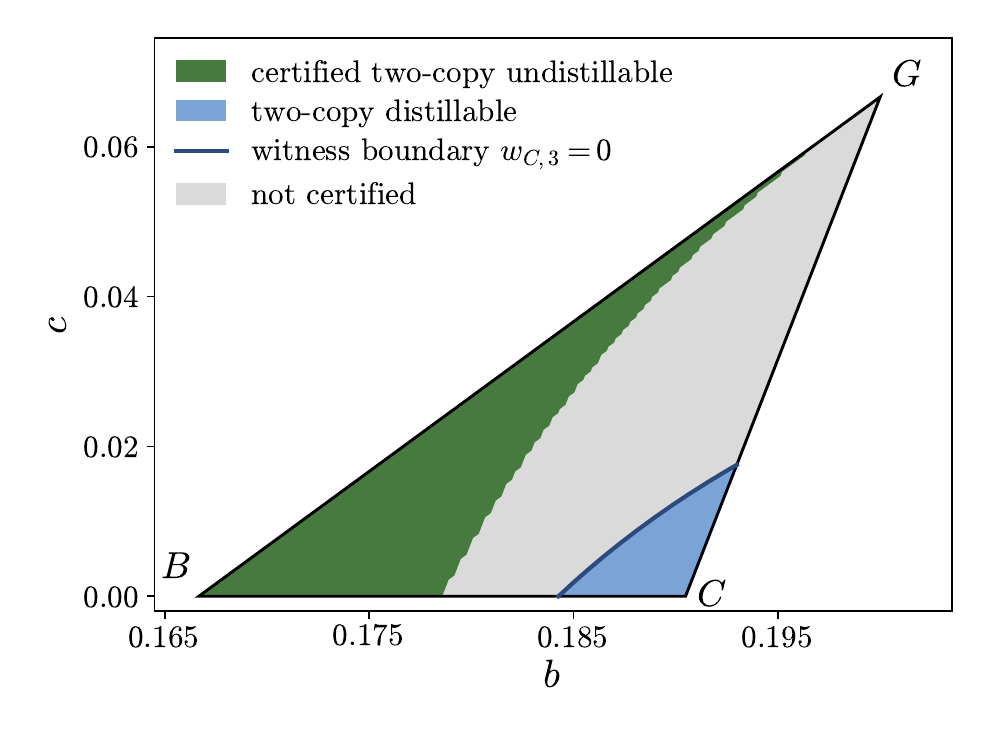}
    \caption{Exactly certified two-copy geometry inside $BCG$ for
    $d=3$.  Green cells are two-copy undistillable by the rational
    decomposition certificates above and occupy exactly $107/300$ of
    $BCG$.  The blue region is the rigorous two-copy-distillable lobe
    $w_{C,3}(b,c)<0$; its dark-blue curve is the exact fixed-witness
    boundary $w_{C,3}=0$.  Gray cells are simply not certified by this
    decomposition and are not asserted to be undistillable.}
    \label{fig:exact-band-d3}
\end{figure}

The ancillary certificate files accompanying this submission consist
of six JSON certificate files and a standalone Python verifier using
only exact rational arithmetic from the standard library. The verifier
reconstructs the building blocks and all six targets from $d$ and the
grid vertices, checks nonnegativity of every certificate weight and
residual coefficient, and verifies the inclusions
$\mathcal S_3\subseteq\cdots\subseteq\mathcal S_8$. It performs
neither linear programming nor floating-point arithmetic. The
numerical generator used to discover these decompositions is not
required for verification.

\section{Point-\texorpdfstring{$C$}{C}-tailored fixed three-copy witnesses and regions}
\label{app:three-copy-fixed}

For the three-copy calculations it is convenient to rewrite
Eq.~\eqref{eq:canonical-PT-general} as
\begin{equation}
    \rho_{b,c}^{\Gamma}
    =
    pI-r\Omega+z\Delta,
    \label{eq:three-copy-prz}
\end{equation}
where
\begin{equation}
    p=\frac{b+c}{2},
    \qquad
    r=\frac{b-c}{2},
    \qquad
    z=
    \frac1d-
    \frac{d-1}{2}b-
    \frac{d+1}{2}c.
    \label{eq:three-copy-prz-definitions}
\end{equation}

\subsection{A zero-margin Schmidt-rank-two direction on
\texorpdfstring{$CG$}{CG}}
\label{app:CG-zero}

Along the line $CG$, the coefficients in
Eq.~\eqref{eq:three-copy-prz} satisfy the exact identity
\begin{equation}
    p-2r+\frac{2}{d}z=0.
    \label{eq:CG-prz-identity}
\end{equation}
Let $P$ be any rank-two projector on $\mathbb C^d$ with constant
diagonal $P_{ii}=2/d$. Such a projector exists for every $d\geq2$;
for example, take the projector onto the span of two columns of the
discrete Fourier matrix. Define
\begin{equation}
    \ket{\zeta_d}
    =
    \frac1{\sqrt2}\operatorname{vec}_{\operatorname r}(P).
    \label{eq:CG-zero-vector}
\end{equation}
Because $\operatorname{rank}P=2$, this vector has Schmidt rank two, and
$\operatorname{tr}P=\operatorname{tr}P^2=2$ gives
\begin{equation}
    \langle I\rangle_{\zeta_d}=1,
    \qquad
    \langle\Omega\rangle_{\zeta_d}=2,
    \qquad
    \langle\Delta\rangle_{\zeta_d}=\frac2d.
\end{equation}
Hence Eqs.~\eqref{eq:three-copy-prz} and
\eqref{eq:CG-prz-identity} imply
\begin{equation}
    \bra{\zeta_d}\rho^\Gamma\ket{\zeta_d}
    =p-2r+\frac2d z=0
\end{equation}
for every $\rho\in CG$. For $k>1$, tensor
$\ket{\zeta_d}$ with arbitrary product vectors on the remaining
copies. The global Schmidt rank remains at most two and the expectation
of $(\rho^\Gamma)^{\otimes k}$ remains zero. This proves the zero-margin 
statement in Eq.~\eqref{eq:CG-zero-main}.

\subsection{A sparse qutrit witness}

On a local three-qutrit space, define
\begin{align}
    \ket{x}
    &=
    \frac{\ket{011}-\ket{102}-\ket{220}}{\sqrt3},
    &
    \ket{y}
    &=
    \frac{\ket{000}+\ket{121}-\ket{212}}{\sqrt3}.
\end{align}
The vectors are orthonormal. Hence
\begin{equation}
    \ket{\Xi_3}
    =
    \frac1{\sqrt2}
    \left(
        \ket{y}_A\ket{y}_B
        -
        \ket{x}_A\ket{x}_B
    \right)
    \label{eq:qutrit-three-copy-witness}
\end{equation}
is normalized and has Schmidt rank exactly two across
$(A_1A_2A_3):(B_1B_2B_3)$. At point $C$,
\begin{equation}
    \bra{\Xi_3}
    \left(\rho_{C,3}^{\Gamma}\right)^{\otimes3}
    \ket{\Xi_3}
    =-\frac1{3087}<0.
    \label{eq:qutrit-three-copy-negative}
\end{equation}

\subsection{A three-line frame witness for
\texorpdfstring{$d\geq4$}{d>=4}}

Let $d\geq4$. For $m\geq3$, introduce the real two-column frame
\begin{equation}
    f_{\ell 1}^{(m)}
    =
    \frac1{\sqrt m}
    \cos\!\left(\frac{2\pi\ell}{m}\right),
    \qquad
    f_{\ell 2}^{(m)}
    =
    \frac1{\sqrt m}
    \sin\!\left(\frac{2\pi\ell}{m}\right),
    \label{eq:three-copy-frame}
\end{equation}
where $\ell=0,\ldots,m-1$. For $m=2$, set
$f_{01}^{(2)}=f_{12}^{(2)}=1/\sqrt2$ and the other two entries to
zero. In either case,
\begin{equation}
    \sum_{\ell=0}^{m-1}
    f_{\ell r}^{(m)}f_{\ell s}^{(m)}
    =
    \frac12\delta_{rs}.
    \label{eq:three-copy-frame-tight}
\end{equation}

Let $I_0$, $I_1$, and $I_2$ be the increasing ordered lists obtained
from $\{0,\ldots,d-1\}$ by deleting, respectively,
$\{1,3\}$, $\{2\}$, and $\{3\}$. Denote their entries by
$i_{a,\ell}$ and define three disjoint coordinate lines
\begin{align}
    \ket{w_{0,\ell}}
    &=
    \ket{i_{0,\ell},2,3},
    &&
    \ell=0,\ldots,d-3,
    \\
    \ket{w_{1,\ell}}
    &=
    \ket{1,i_{1,\ell},3},
    &&
    \ell=0,\ldots,d-2,
    \\
    \ket{w_{2,\ell}}
    &=
    \ket{3,2,i_{2,\ell}},
    &&
    \ell=0,\ldots,d-2.
\end{align}
For $s=1,2$, put
\begin{align}
    \ket{L_{0,s}}
    &=
    \sum_{\ell=0}^{d-3}
    f_{\ell s}^{(d-2)}\ket{w_{0,\ell}},
    \\
    \ket{L_{a,s}}
    &=
    \sum_{\ell=0}^{d-2}
    f_{\ell s}^{(d-1)}\ket{w_{a,\ell}},
    \qquad a=1,2.
\end{align}
The disjoint supports and Eq.~\eqref{eq:three-copy-frame-tight} imply
\begin{equation}
    \braket{L_{a,r}|L_{b,s}}
    =
    \frac12\delta_{ab}\delta_{rs}.
    \label{eq:three-line-orthogonality}
\end{equation}

For $0<\lambda<2$, let $\mu=1-\lambda/2$ and define
\begin{align}
    \ket{u_s(\lambda)}
    &=
    \sqrt\lambda\ket{L_{0,s}}
    +\sqrt\mu\ket{L_{1,s}}
    +\sqrt\mu\ket{L_{2,s}},
    \\
    \ket{v_s(\lambda)}
    &=
    -\sqrt\lambda\ket{L_{0,s}}
    +\sqrt\mu\ket{L_{1,s}}
    +\sqrt\mu\ket{L_{2,s}}.
\end{align}
Equation~\eqref{eq:three-line-orthogonality} gives
$\braket{u_r|u_s}=\braket{v_r|v_s}=\delta_{rs}$. Therefore
\begin{equation}
    \ket{\Theta_{d,3}(\lambda)}
    =
    \frac1{\sqrt2}
    \sum_{s=1}^{2}
    \ket{u_s(\lambda)}_A\ket{v_s(\lambda)}_B
    \label{eq:three-line-witness}
\end{equation}
is normalized and has Schmidt rank exactly two.

Direct evaluation at point $C$ gives
\begin{equation}
    \begin{aligned}
    &\bra{\Theta_{d,3}(\lambda)}
    \left(\rho_{C,d}^{\Gamma}\right)^{\otimes3}
    \ket{\Theta_{d,3}(\lambda)}
    \\
    &\quad=
    \frac{1}{d^3(d-1)(3d-2)^3}
    \left[
        \frac{P_d}{4(d-2)}\lambda^2
        -N_d\lambda
        +(d-2)^2
    \right],
    \label{eq:three-line-pointC-lambda}
    \end{aligned}
\end{equation}
where
\begin{equation}
    P_d=13d^3-34d^2+28d-8,
    \qquad
    N_d=5d^2-8d+4.
\end{equation}
The minimizing weight is
\begin{equation}
    \lambda_d
    =
    \frac{2(d-2)(5d^2-8d+4)}
    {P_d}.
    \label{eq:lambda-d-three-copy}
\end{equation}
Writing
$\ket{\Theta_{d,3}}=\ket{\Theta_{d,3}(\lambda_d)}$, we obtain
\begin{equation}
    \bra{\Theta_{d,3}}
    \left(\rho_{C,d}^{\Gamma}\right)^{\otimes3}
    \ket{\Theta_{d,3}}
    =
    -\frac{4(d-2)}
    {d(3d-2)^2 P_d}<0.
    \label{eq:three-line-pointC-negative}
\end{equation}
For every $d\geq4$, one has $0<\lambda_d<2$, so this choice is
admissible.

\subsection{The qutrit cubic}

For the fixed vector $\ket{\Xi_3}$ in
Eq.~\eqref{eq:qutrit-three-copy-witness}, direct evaluation gives
\begin{equation}
    \boxed{
    \begin{aligned}
    w_{\Xi,3}^{(3)}(b,c)
    &:={}
    \bra{\Xi_3}
    \left(\rho_{b,c}^{\Gamma}\right)^{\otimes3}
    \ket{\Xi_3}
    \\
    &={}
    -\frac1{324}\Big(
        27b^3+243b^2c-81b^2
        +405bc^2
        % \\
        % &\hspace{20mm}
        -270bc+36b
        -27c^3-81c^2+36c-4
    \Big).
    \end{aligned}
    }
    \label{eq:qutrit-three-copy-plane}
\end{equation}
Every physical qutrit canonical state with
$w_{\Xi,3}^{(3)}(b,c)<0$ is three-copy distillable. Along $CB$, the
expectation factorizes as
\begin{equation}
    w_{\Xi,CB}^{(3)}(t)
    =
    \frac{(t+6)(t^2+96t-48)}{889056},
\end{equation}
and is negative for
\begin{equation}
    0\leq t<28\sqrt3-48.
\end{equation}
Along $CG$,
\begin{equation}
    w_{\Xi,CG}^{(3)}(t)
    =
    -\frac{(t-15)(2t+5)(19t-5)}{1157625},
\end{equation}
so the first zero is $t=5/19$, the same cutoff as the fixed two-copy
qutrit witness. In fact, combining this expression with
Eq.~\eqref{eq:general-CG-witness} at $d=3$ gives the exact relation
\begin{equation}
    w_{\Xi,3}^{(3)}\!\left(\rho_{CG}(t)\right)
    =
    \frac{2(15-t)}{105}\,
    w_{C,3}\!\left(\rho_{CG}(t)\right),
    \label{eq:qutrit-CG-two-three-factor}
\end{equation}
whose prefactor is positive throughout $0\leq t\leq1$. Thus the common
fixed-witness cutoff is algebraic rather than a numerical coincidence.

\subsection{The cubic for the three-line witness}

For $d\geq4$, let
\begin{equation}
    w_{d,\lambda}^{(3)}(b,c)
    :=
    \bra{\Theta_{d,3}(\lambda)}
    \left(\rho_{b,c}^{\Gamma}\right)^{\otimes3}
    \ket{\Theta_{d,3}(\lambda)}.
\end{equation}
Expanding Eq.~\eqref{eq:three-copy-prz} over the three copies and using
the disjoint line supports together with
Eq.~\eqref{eq:three-copy-frame-tight} gives
\begin{equation}
    \boxed{
    \begin{aligned}
    w_{d,\lambda}^{(3)}(b,c)
    ={}&
    p^3+A_{d,\lambda}p^2z-B_{d,\lambda}p^2r
    +C_{d,\lambda}pz^2
    -D_{d,\lambda}pzr
    \\
    &+E_{d,\lambda}pr^2
    +F_{d,\lambda}z^3
        -G_{d,\lambda}z^2r
    +H_{d,\lambda}zr^2
    -J_{d,\lambda}r^3,
    \end{aligned}
    }
    \label{eq:three-line-plane-compact}
\end{equation}
where
\begin{align}
    A_{d,\lambda}
    &={}
    \frac{
        (d^2-2)\lambda^2
        -4(d-2)\lambda
        +4d(d-2)
    }{4(d-2)(d-1)},
    \\
    B_{d,\lambda}
    &={}
    \lambda^2-\lambda+2,
    \\
    C_{d,\lambda}
    &={}
    \frac{
        (3d^2+3d-10)\lambda^2
        -4(d^2+d-6)\lambda
        +4(d^2+d-6)
    }{8(d-2)(d-1)},
    \\
    D_{d,\lambda}
    &={}
    \frac{
        (9d^2-21d+10)\lambda^2
        -4(3d^2-7d+2)\lambda
        +4(3d^2-7d+2)
    }{4(d-2)(d-1)},
    \\
    E_{d,\lambda}
    &={}
    \frac{23\lambda^2-36\lambda+20}{8},
    \\
    F_{d,\lambda}
    &={}
    \frac{
        (3d-4)\lambda^2
        -4(d-2)\lambda
        +4(d-2)
    }{4(d-2)(d-1)},
    \\
    G_{d,\lambda}
    &={}
    \frac{
        (3d^2-3d-2)\lambda^2
        -4(d^2-d-2)\lambda
        +4(d^2-d-2)
    }{4(d-2)(d-1)},
    \\
    H_{d,\lambda}
    &={}
    \frac{
        (10d^2-27d+16)\lambda^2
        +(-16d^2+44d-24)\lambda
        +8d^2-20d+8
    }{4(d-2)(d-1)},
    \\
    J_{d,\lambda}
    &={}
    2(\lambda-1)^2.
\end{align}
Substituting the point-$C$-optimized value in
Eq.~\eqref{eq:lambda-d-three-copy} defines the fixed cubic
\begin{equation}
    w_{\Theta,d}^{(3)}(b,c)
    =
    w_{d,\lambda_d}^{(3)}(b,c).
    \label{eq:fixed-three-copy-plane-general}
\end{equation}
The full region certified by the displayed constructions at three
copies is
\begin{equation}
    \mathcal R_{d,\mathrm{cert}}^{(3)}
    =
    \left\{w_{C,d}<0\right\}
    \cup
    \begin{cases}
        \left\{w_{\Xi,3}^{(3)}<0\right\},&d=3,\\
        \left\{w_{\Theta,d}^{(3)}<0\right\},&d\geq4,
    \end{cases}
    \label{eq:fixed-three-copy-certified-union}
\end{equation}
where all sets are understood as subsets of the physical canonical plane.
The first set is included because every two-copy certificate also
yields a three-copy certificate by tensoring the witness with a
suitable product vector on the third copy.

\begin{figure*}[!t]
\centering
\begin{tikzpicture}[
    x=5.1cm,
    y=3.7cm,
    every node/.style={font=\small},
    point/.style={circle,fill=black,inner sep=1.3pt}
]
    % Panel (a): d=3. Coordinates are an affine schematic of BCG.
    \begin{scope}[xshift=0cm]
        \coordinate (B3) at (0,0);
        \coordinate (C3) at (1,0);
        \coordinate (G3) at (0.78,0.78);

        \path[fill=orange!20]
            (C3)--
            plot[smooth] coordinates {
                (0.5026,0.0000) (0.5259,0.0243)
                (0.5535,0.0482) (0.5848,0.0711)
                (0.6190,0.0923) (0.6547,0.1113)
                (0.6906,0.1279) (0.7254,0.1423)
                (0.7584,0.1545) (0.7892,0.1648)
                (0.8176,0.1735) (0.8435,0.1810)
                (0.8671,0.1874) (0.8886,0.1928)
                (0.9081,0.1975) (0.9259,0.2017)
                (0.9421,0.2053)
            }--cycle;

        \path[fill=blue!18]
            (C3)--
            plot[smooth] coordinates {
                (0.7390,0.0000) (0.7471,0.0130)
                (0.7559,0.0264) (0.7656,0.0401)
                (0.7762,0.0542) (0.7876,0.0685)
                (0.7998,0.0828) (0.8128,0.0970)
                (0.8264,0.1110) (0.8405,0.1247)
                (0.8549,0.1380) (0.8696,0.1508)
                (0.8844,0.1629) (0.8992,0.1745)
                (0.9137,0.1854) (0.9281,0.1957)
                (0.9421,0.2053)
            }--cycle;

        \draw[thick] (B3)--(C3)--(G3)--cycle;
        \draw[orange!85!black,thick,dashed]
            plot[smooth] coordinates {
                (0.5026,0.0000) (0.5259,0.0243)
                (0.5535,0.0482) (0.5848,0.0711)
                (0.6190,0.0923) (0.6547,0.1113)
                (0.6906,0.1279) (0.7254,0.1423)
                (0.7584,0.1545) (0.7892,0.1648)
                (0.8176,0.1735) (0.8435,0.1810)
                (0.8671,0.1874) (0.8886,0.1928)
                (0.9081,0.1975) (0.9259,0.2017)
                (0.9421,0.2053)
            };
        \draw[blue!70!black,thick]
            plot[smooth] coordinates {
                (0.7390,0.0000) (0.7471,0.0130)
                (0.7559,0.0264) (0.7656,0.0401)
                (0.7762,0.0542) (0.7876,0.0685)
                (0.7998,0.0828) (0.8128,0.0970)
                (0.8264,0.1110) (0.8405,0.1247)
                (0.8549,0.1380) (0.8696,0.1508)
                (0.8844,0.1629) (0.8992,0.1745)
                (0.9137,0.1854) (0.9281,0.1957)
                (0.9421,0.2053)
            };

        \node[point,label=below left:$B$] at (B3) {};
        \node[point,label=below:$C$] at (C3) {};
        \node[point,label=right:$G$] at (G3) {};
        \node at (0.50,0.88) {(a) $d=3$};
    \end{scope}

    % Panel (b): d=4, representative of the d>=4 fixed construction.
    \begin{scope}[xshift=6.0cm]
        \coordinate (B4) at (0,0);
        \coordinate (C4) at (1,0);
        \coordinate (G4) at (0.78,0.78);

        \path[fill=orange!20]
            (C4)--
            plot[smooth] coordinates {
                (0.7037,0.0000) (0.7106,0.0148)
                (0.7188,0.0304) (0.7283,0.0465)
                (0.7392,0.0632) (0.7516,0.0801)
                (0.7654,0.0970) (0.7806,0.1136)
                (0.7970,0.1298) (0.8143,0.1452)
                (0.8322,0.1596) (0.8503,0.1731)
                (0.8685,0.1854) (0.8845,0.1998)
                (0.8992,0.2167) (0.9144,0.2329)
                (0.9300,0.2482)
            }--cycle;

        \path[fill=blue!18]
            (C4)--
            plot[smooth] coordinates {
                (0.7846,0.0000) (0.7875,0.0109)
                (0.7908,0.0226) (0.7946,0.0352)
                (0.7991,0.0487) (0.8044,0.0630)
                (0.8106,0.0783) (0.8178,0.0944)
                (0.8260,0.1112) (0.8354,0.1286)
                (0.8460,0.1465) (0.8578,0.1644)
                (0.8707,0.1823) (0.8845,0.1998)
                (0.8992,0.2167) (0.9144,0.2329)
                (0.9300,0.2482)
            }--cycle;

        \draw[thick] (B4)--(C4)--(G4)--cycle;
        \draw[orange!85!black,thick,dashed]
            plot[smooth] coordinates {
                (0.7037,0.0000) (0.7106,0.0148)
                (0.7188,0.0304) (0.7283,0.0465)
                (0.7392,0.0632) (0.7516,0.0801)
                (0.7654,0.0970) (0.7806,0.1136)
                (0.7970,0.1298) (0.8143,0.1452)
                (0.8322,0.1596) (0.8503,0.1731)
                (0.8685,0.1854) (0.8845,0.1998)
                (0.8992,0.2167) (0.9144,0.2329)
                (0.9300,0.2482)
            };
        \draw[blue!70!black,thick]
            plot[smooth] coordinates {
                (0.7846,0.0000) (0.7875,0.0109)
                (0.7908,0.0226) (0.7946,0.0352)
                (0.7991,0.0487) (0.8044,0.0630)
                (0.8106,0.0783) (0.8178,0.0944)
                (0.8260,0.1112) (0.8354,0.1286)
                (0.8460,0.1465) (0.8578,0.1644)
                (0.8707,0.1823) (0.8845,0.1998)
                (0.8992,0.2167) (0.9144,0.2329)
                (0.9300,0.2482)
            };

        \node[point,label=below left:$B$] at (B4) {};
        \node[point,label=below:$C$] at (C4) {};
        \node[point,label=right:$G$] at (G4) {};
        \node at (0.50,0.88) {(b) $d=4$};
    \end{scope}

    \node[align=left,text=blue!70!black,font=\footnotesize]
        at (1.95,-0.16)
        {fixed two-copy certificate};
    \node[align=left,text=orange!85!black,font=\footnotesize]
        at (1.95,-0.27)
        {additional fixed three-copy certificate};

\end{tikzpicture}
\caption{
Fixed-witness inner bounds within $BCG$. Blue denotes the
region certified by the fixed two-copy witness $\ket{\Psi_d}$. Orange
denotes the additional region reached by the point-$C$-tailored fixed
three-copy witness; the full set certified at three copies by the
displayed constructions is the union of the blue and orange parts. For
$d=3$, the qutrit witness $\ket{\Xi_3}$ enlarges the certified lobe
throughout the displayed range and shares the same cutoff on $CG$. For
$d\geq4$, represented by $d=4$ in panel (b), the witness
$\ket{\Theta_{d,3}}$ gives an additional bulge mainly toward $B$ and
the interior, while the fixed two-copy certificate remains stronger
near $CG$. The orange extension improves the available certificate but does not
by itself establish two-copy undistillability. Only analytic fixed
certificates are shown; the panels use affine schematic coordinates
and are not on a common physical scale.
}
\label{fig:fixed-three-copy-schematic}
\end{figure*}

For orientation, the first boundary intersections along $CB$ and $CG$
are
\begin{equation}
\begin{array}{c|cc|cc}
 d
 &t_{B}^{(2)}&t_{B}^{(3),\mathrm{cert}}
 &t_{G}^{(2)}&t_{G}^{(3),\mathrm{cert}}
 \\
 \hline
 3&0.260990&0.497423&0.263158&0.263158\\
 4&0.215352&0.296299&0.318182&0.318182\\
 5&0.183838&0.251100&0.341772&0.341772\\
 6&0.160645&0.218640&0.354839&0.354839
\end{array}
\label{eq:fixed-three-copy-cutoff-table}
\end{equation}
For $d\geq4$, the fixed three-line witness is slightly weaker than the
two-copy witness along $CG$; the certified union therefore retains the
larger two-copy cutoff in that direction. These values concern only the
analytic fixed witnesses and are not claims about the fully optimized
three-copy-distillable region.

\section{Numerical searches for the finite-copy boundary}
\label{app:numerics}

To explore whether the analytic fixed witnesses are close to the
optimized Schmidt-rank-two boundary, we numerically minimize
\begin{equation}
    m_k^{(d)}(b,c)
    =
    \min_{\substack{\|\psi\|=1\\ \SR(\psi)\leq2}}
    \bra{\psi}
    (\rho_{b,c}^{\Gamma})^{\otimes k}
    \ket{\psi}.
    \label{eq:numerical-mk}
\end{equation}
The search alternates the optimal two-dimensional local subspaces on
Alice and Bob and uses many random restarts, with the analytic witness
included as a seed where appropriate. Because this is a nonconvex
alternating optimization, the returned value $\widehat m_k$ is an
\emph{upper} bound on the true minimum $m_k^{(d)}$. Thus a negative
value certifies a Schmidt-rank-two distillation direction, whereas a
nonnegative value is evidence only and cannot certify
undistillability.

For two copies in $d=3$, a 231-point grid over $BCG$ contains 21
sampled points inside the analytic lobe. At all 21 points the best
search value agrees with $w_{C,3}$ to relative precision of order
$10^{-13}$, while none of the other 210 points yields a negative value.
Across five rays from $C$, evaluations at $0.97$ and $1.03$ times the
analytic cutoff give the same local boundary, except along $CG$ where
the search value becomes zero after the cutoff. In $d=4$, an analogous
66-point grid gives agreement at all 8 sampled lobe points, with no
negative value at the remaining 58 points, and the boundary-ray tests
show the same behavior. These data motivate the low-dimensional
conjecture that for $d=3,4$ the fixed lobe is exactly the
 two-copy-distillable subset of $BCG$.

Targeted higher-dimensional tests show that this sharpness does not
persist uniformly in $d$. The sampled tests for $d=5$ and the $CB$
boundary tests for $d=6$ remain consistent with the analytic cutoff.
For $d=7$, however, a search at $1.05\,t_B^{(7)}$, where the fixed
witness is already positive, finds a Schmidt-rank-two value about
$-5.9\times10^{-7}$. For $d=8$ a negative value about
$-3.7\times10^{-8}$ is still found at $1.03\,t_B^{(8)}$, although no
negative value was found in the lighter search by $1.15\,t_B^{(8)}$.
The improved vectors retain two equal Schmidt coefficients but do not
use the single distinguished-label structure of
Eq.~\eqref{eq:general-witness}. Thus a different family of
Schmidt-rank-two tests appears to take over between $d=6$ and $d=7$.

At three copies for $d=3$, the sparse witness $\ket{\Xi_3}$ is recovered
at $C$ under 120 random restarts. A coarser 66-point grid over $BCG$
finds no negative direction among the sampled points outside the union
of the analytic two-copy and three-copy certified regions. Along $CB$
the search follows the $\ket{\Xi_3}$ branch through its cutoff
$28\sqrt3-48$, while along $CG$ it returns zero beyond $5/19$ and finds
no negative direction there. These three-copy searches used fewer
restarts and the landscape is more rugged, so they should be regarded
as weaker evidence than the two-copy scans.

These numerical observations are not used as undistillability proofs.
Their role is only to assess how closely the fixed analytic witnesses
track the optimized Schmidt-rank-two boundary. In particular, they
show that the equal-norm-frame witness can be nearly sharp in low
dimensions without being uniformly optimal as $d$ increases.

\end{document}